\documentclass[%
10pt,tightenlines,
twocolumn,
superscriptaddress,
 amsmath,amssymb,
 aps,
pra,
]{revtex4-2}

\usepackage[margin=0.8in]{geometry}

\usepackage{graphicx}
\usepackage{dcolumn}
\usepackage{bm}
\usepackage{color,xcolor}
\usepackage[colorlinks]{hyperref}
\usepackage{CJKutf8}
\usepackage{amsthm}
\usepackage{amsmath,amssymb}
\usepackage{braket}
\usepackage{xcolor}
\usepackage{esvect}

\usepackage{algorithm}
\usepackage{algpseudocode}   
\usepackage{mathtools}
\usepackage{graphicx}
\usepackage{caption}
\usepackage{subcaption}
\usepackage{pifont}
\newcommand{\yes}{\ding{51}} 
\newcommand{\no}{\ding{55}}  
\newtheorem{theorem}{Theorem}[section]
\newtheorem*{theorem*}{Theorem}
\newtheorem{lemma}[theorem]{Lemma}
\newtheorem{definition}[theorem]{Definition}
\usepackage{bbm}

\newcommand{\hc}{\hat{c}}
\renewcommand{\hm}{\hat{m}}

\usepackage[justification=raggedright,singlelinecheck=false]{caption}
\newcommand{\var}{\text{Var}}

\usepackage{makecell}
\usepackage{enumitem}
\usepackage{cancel}

\begin{document}
\begin{CJK}{UTF8}{gbsn}

\title{Optimal and Robust In-situ Quantum Hamiltonian Learning through Parallelization}

\author{Suying Liu}
\email[Electronic address: ]{syliu@umd.edu}
\affiliation{Joint Center for Quantum Information and Computer Science, University of Maryland, College Park, Maryland 20742, USA}
\affiliation{Department of Computer Science, University of Maryland, College Park, USA}

\author{Xiaodi Wu} 
\affiliation{Joint Center for Quantum Information and Computer Science, University of Maryland, College Park, Maryland 20742, USA}
\affiliation{Department of Computer Science, University of Maryland, College Park, USA}

\author{Murphy Yuezhen Niu}
\email[Electronic address: ]{murphyniu@ucsb.edu}
\affiliation{Department of Computer Science, University of California, Santa Barbara, Santa Barbara, CA 93106, USA}
\affiliation{Google Quantum AI, Venice, California, CA 90291, USA}

\date{\today}
\begin{abstract}
Hamiltonian learning is a cornerstone for advancing accurate many-body simulations, improving quantum device performance, and enabling quantum-enhanced sensing. Existing readily deployable quantum metrology techniques primarily focus on achieving Heisenberg-limited precision in one- or two-qubit systems. In contrast, general Hamiltonian learning theories address broader classes of unknown Hamiltonian models but are highly inefficient due to the absence of prior knowledge about the Hamiltonian. There remains a lack of efficient and practically realizable Hamiltonian learning algorithms that directly exploit the known structure and prior information of the Hamiltonian, which are typically available for a given quantum computing platform. In this work, we present the first Hamiltonian learning algorithm that achieves both Cramér–Rao lower bound saturated optimal precision and robustness to realistic noise, while exploiting device structure for quadratic reduction in experimental cost for fully connected Hamiltonians. Moreover, this approach enables simultaneous in-situ estimation of all Hamiltonian parameters without requiring the decoupling of non-learnable interactions during the same experiment, thereby allowing comprehensive characterization of the system’s intrinsic contextual errors. Notably, our algorithm does not require deep circuits and remains robust against both depolarizing noise and time-dependent coherent errors. We demonstrate its effectiveness with a detailed experimental proposal along with supporting numerical simulations on Rydberg atom quantum simulators, showcasing its potential for high-precision Hamiltonian learning in the NISQ era.
\end{abstract}

\maketitle
\end{CJK}
\section{\label{sec:introduction} Introduction}
Learning the underlying Hamiltonian of a quantum system is a foundational problem with broad implications across quantum simulation, quantum control, and fault-tolerant quantum computing. Accurate Hamiltonian characterization enables high-fidelity simulation of quantum many-body dynamics \cite{dag2024distanceerror,arute2020periodiccalibration,andersen2025thermalization,karamlou2024probing,manovitz2025quantum,bluvstein2024logical}, facilitates the calibration and benchmarking of quantum hardware \cite{shulman2014suppressing,niu2019universal,hangleiter2024robustly,olsacher2025hamiltonian}, and underpins the demonstration of quantum advantage in scalable fault-tolerant quantum computers \cite{de2005quantum,leibfried2004toward,mckenzie2002experimental,allen2025sensing}. In this context, parallel in-situ Hamiltonian learning has emerged as a critical component as the size and model complexity of current quantum devices increase dramatically for achieving fault-tolerant computing and completing practical applications. Parallelization means we can learn multiple Hamiltonian parameters with the same experiment, which is essential for efficiently processing the exponentially large Hilbert spaces encountered in large-scale quantum systems \cite{preskill1998fault}. In-situ learning enables characterization of specific Hamiltonian parameters without requiring the deactivation of other couplings. It is needed for characterizing contextual and correlated errors, such as quantum crosstalk, that cannot be captured by ex-situ methods. These capabilities make parallel in-situ learning indispensable for reliable and scalable quantum computation. 

Parallel and in-situ Hamiltonian learning are crucial for resolving key challenges for quantum simulators. For example, in Rydberg atom devices, atom position errors, caused by turning off optical traps during evolution, significantly impact simulation accuracy as the interaction strength scales with distance as $c_{ij} = C_6/R^6_{ij}$. This has led to notable discrepancies between theoretical predictions and experimental results in many-body dynamics \cite{dag2024distanceerror}. A first-principles understanding of these mismatches via in-situ Hamiltonian learning is required to bridge the gap and validate error models for noisy quantum simulators, paving the way for the study of complex quantum many-body systems in the NISQ era. 
\begin{table*}[t]
    \centering
    \begin{tabular}{|c|c|c|c|c|c|}
    \hline
          & Heisenberg limit & Robust & In-situ & $\#$ experiment rounds & Applicability \\
    \hline
    \cite{kimmel2015rpe}  & \yes & \no & - & - &  digital\\
     \hline 
     \cite{dong2025qspe} & \yes &\yes& - & - & digital\\
     \hline 
      \cite{huang2023learning}&  \yes & \no & \no & $O(n^2)$& hybrid \\
    \hline
    This work [Algorithm \ref{alg:qspe-parallel}] &\yes &\yes& \yes & $O(n)$ & digital and hybrid\\
    \hline
    This work [Algorithm \ref{alg:qspe-analog}] &\yes &\yes& \yes & $O(n^2)$ & analog\\
    \hline
    \end{tabular}
    \caption{Comparison of this work and previous methods: (i) ability to reach the Heisenberg limit and robustness against SPAM, decoherence, and coherent errors; (ii) performance in learning fully connected Hamiltonians, including in-situ learning ability, required experiment rounds, and quantum device requirement for implementing the algorithm. `-' represents the method is not applicable at this scenario.}

    \label{tab:table of comparison}
\end{table*}\\
Previous works on Hamiltonian learning~\cite{huang2023learning,hu2025ansatz,ma2024learning,li2024heisenberg,ni2024quantum} have yet to demonstrate parallel in-situ learning of many-body Hamiltonians and have not provided readily deployable protocols that maintain robustness under realistic noise conditions. Instead, these studies have primarily focused on establishing fundamental bounds for general learning tasks, without offering protocols that are directly applicable to current quantum devices. Existing quantum hardware suffers from diverse noise sources—including decoherence, systematic coherent errors, and time-dependent drift \cite{manuel-endres-2024benchmarking}, which severely degrade the performance of learning proposals that are only robust to state preparation and measurement (SPAM) errors. Further, this line of approaches depends on the Hamiltonian Reshaping \cite{huang2023learning} technique, which decomposes the many-body Hamiltonian into non-interacting clusters. As a result, these methods do not support in-situ Hamiltonian learning in general scenarios. In practice, several works in quantum metrology and quantum calibration have addressed the practical task of learning Hamiltonian coefficients as a means to certify quantum simulators \cite{andersen2025thermalization,dag2024distanceerror,bernien2017probing,kimmel2015rpe,dong2025qspe}. However, these protocols are limited to learning two-qubit systems. For instance, the calibration method in \cite{andersen2025thermalization} estimates coupling strengths between two superconducting qubits using pairwise measurements. Similarly, in \cite{bernien2017probing}, interaction terms are inferred through a four-photon process involving only two neighboring atoms. As a consequence, these methods are inherently limited to learning a single interaction term per qubit at a time. Therefore, without parallelization, $O(k)$ experiments are required to characterize couplings associated with a $k$-connected single qubit. Furthermore, because these approaches decompose the original Hamiltonian into pairwise interactions, they preclude in-situ learning and thus fail to capture contextual errors, such as unwanted correlation errors that emerge only when multiple qubits interact simultaneously. Despite these rapid advances, the realization of parallel and in-situ quantum computation remains an open challenge, both in theory and in practice.  

In this work, we introduce the first in-situ Hamiltonian learning algorithm that incorporates both the structural prior of realistic systems and noise resilience. Specifically, we target a class of parallel-learnable Hamiltonians that cover a wide range of physical devices (Rydberg atom, trapped ion, and superconducting qubits), where subsets of parameters can be estimated simultaneously in-situ. The parallelization feature of our algorithm achieves a quadratic reduction in experimental cost with respect to the number of qubits, compared to previous approaches for estimating fully connected Hamiltonians. More specifically, our algorithm leverages parallelization to solve any set of $k$ pairwise coupling terms that share the same qubit using only $O(1)$ experiments. Besides, our algorithm leverages Quantum Signal-Processing (QSP) to decouple the parameters, enabling robustness against SPAM, decoherence, and time-dependent coherent errors, and achieving optimal precision. The algorithm’s optimality and robustness make it well-suited for learning device Hamiltonians on current quantum hardware. Moreover, the in situ feature of the learning algorithm allows us to fill in the blank in capturing the Rydberg atom distance error in many-body systems, serving as the first principles approach to validate the error model in theory and practice. To demonstrate its practicality, we present a detailed experimental proposal along with supporting numerical simulations for high-precision atom-distance estimation on Rydberg atom quantum simulators. A comprehensive comparison between our learning algorithms and previous works on learning Hamiltonian parameters is shown in Table \ref{tab:table of comparison}.

The rest of the paper is structured as follows. In Sec.~\ref{subsec:problem setting}, we define the parallel-learnable Hamiltonian and introduce specific Problem \ref{def:prb 1} as the main focus for presenting our techniques. The main results of the work are summarized in Sec.~\ref{sec:main results}. Sec.~\ref{subsec:methods} presents our learning algorithm with the introduction of the QSPE-based algorithm for two-atom learning and followed by the general parallel learning algorithm for $n$ atom case. Sec.~\ref{subsec:optimality discussion} discusses the optimality of our algorithm with matching Cramér–Rao bound and computed variance. In Sec.~\ref{subsec:robustness}, we show the robustness of the algorithm against decoherence, coherent state preparation error, readout error, and time-dependent coherent error. In Sec.~\ref{sec:experiment proposal}, we present a detailed experimental proposal with corresponding parameter settings for the application of the algorithm on current Rydberg atom quantum simulators. We conclude in Sec.~\ref{sec:discussion} with open questions for further exploration.

\section{\label{sec:results} Results}
\subsection{\label{subsec:problem setting}Problem definition}
Given that current quantum devices operate under well-defined Hamiltonian structures, accurately characterizing the Hamiltonian coefficients is essential for fully utilizing these systems. In this work, we address the problem of estimating the parameters of fixed-term Hamiltonians, where the structure of the terms is known, but the corresponding coefficients are unknown. Our learning technique applies to any general parallel-learnable Hamiltonian, defined as follows.
\begin{definition}\label{def:pl H}
    The Hamiltonian is parallel-learnable if it can be written as a direct sum of $ 2^{n-1}$ non-diagonal Hermitian matrices with dimension $2 \times 2$ up to a realizable unitary transformation. That is, for any Hamiltonian $H$ which is parallel-learnable, there exists a realizable unitary transformation $U$, such that
    \begin{equation}\label{eqn:pl H}
        U^{\dagger}HU = \bigoplus_{i = 1}^{2^{n-1}}A_i,
    \end{equation}
    where $A_i$'s are $2 \times 2$ non-diagonal Hermitian matrices. 
\end{definition}

From the definition, we note that the parallel-learnable structure requires the Hamiltonian $H$ to have $2^{n-1}$ invariant subspaces under certain bases. The unitary transformation $U$ is used to change it from the computational basis to which the block-diagonal structure is presented. In principle, any Hamiltonian can be transformed into a direct sum of non-diagonal \( 2 \times 2 \) blocks through a specially designed unitary transformation \( U \). A specific example of such a transformation is \( U := R V \), where \( V \) is the eigenbasis matrix that diagonalizes the Hamiltonian, and \( R \) is a block-wise rotation that transforms each \( 2 \times 2 \) diagonal block into a non-diagonal form. (See Supplementary Material Sec.~\ref{SM:pl hamiltonian def} for more details on this construction for arbitrary Hamiltonians.) However, the unitary transformation utilized here is constructed from the eigenbasis matrix $V$ of the Hamiltonian $H$, which is not always feasible to implement on current quantum hardware. Therefore, in this work, we focus on Hamiltonians that preserve the parallel-learnable structure, allowing for physically implementable unitary transformations. We present several illustrative examples of parallel-learnable Hamiltonians and their corresponding unitary transformations in the Supplementary Material Sec.~\ref{SM:pl hamiltonian def}. We notice from those examples that parallel-learnable Hamiltonians frequently appear in many-body systems, particularly in the design of foundational quantum platforms. Examples include the transverse field Ising model, which captures the dynamics of Rydberg interactions \cite{bluvstein2021controlling,browaeys2020many}, and the XY model with local 
$Z$ fields, commonly used to approximate superconducting qubit couplings and interactions in trapped ion gates~\cite{kjaergaard2020superconducting,siddiqi2021engineering,msgate}. Here, we mainly focus on the following Hamiltonian, targeting learning all the Hamiltonain coefficients in-situ and in parallel. 
\begin{definition}[Problem 1]\label{def:prb 1}
Consider the Hamiltonian of $n$-atom Rydberg arrays, \begin{equation}\label{eqn: all2all H}
    H = \sum_{\substack{i,j \in\{1,\cdots,n\} \\ i\neq j}} c_{ij} Z_iZ_j,
\end{equation}
which models pure atom-atom interaction between any two Rydberg atoms. There are in total $\frac{n(n-1)}{2}$ terms in the Hamiltonian, whose coefficients are characterized by the parameter vector  $\mathbf{c} = [c_{11},c_{12},\cdots,c_{n(n-1)}]^T$. We aim to construct estimators $\mathbf{\hc}=[\hat{c}_{11},\hc_{12},\cdots,\hc_{n(n-1)}]^T$ for parameter vector $\mathbf{c}$. 
\end{definition}

This problem is motivated by a fundamental challenge from the Rydberg-array-based quantum computers. The Rydberg atom quantum device is built on neutral atoms, which are trapped with optical tweezers \cite{morgado2021quantum}. But those traps are turned off during the evolution phase in quantum computing \cite{bluvstein2021controlling,browaeys2020many}, resulting in unavoidable atom position movements. Given that the interaction strength depends inversely on the atom distance, i.e.,$c_{ij} = C_6/R^6_{ij}$, small deviations in interatomic distances can lead to sixth-order amplification of errors in the Hamiltonian parameters. This amplification can cause significant discrepancies between theory and experiment \cite{dag2024distanceerror}. So it is critical to measure the atom-atom interaction inside the system Hamiltonian to ensure the accuracy of the quantum simulation. However, the pure interaction Hamiltonian in Eq.~\eqref{eqn: all2all H} is already diagonal, thus does not directly satisfy the definition of parallel-learnable Hamiltonian discussed in Definition \ref{def:pl H}. Therefore, instead of tackling the pure interaction Hamiltonian, we consider the full Hamiltonian, which is constructed by adding a single $X_i$ term to the target Hamiltonian $H$.
\begin{definition}
    The full Hamiltonian of the pure interaction Hamiltonian given in Eq.~\eqref{eqn: all2all H} is
    \begin{equation}\label{eqn:full H}
         H = a_iX_i+\sum_{\substack{i,j \in\{1,\cdots,n\} \\ i\neq j}} c_{ij} Z_iZ_j,
    \end{equation}
     where $X_i$ denotes the application of $X$ to the $i$th atom and $a_i$ is the parameter of $X_i$ term.
\end{definition}
We note here that there are $n$ different full Hamiltonians as there are $n$ atoms inside the system. We will utilize $(n-1)$ of them for learning the full parameter vectors later. For each full Hamiltonian corresponding to a different \( X_i \) term, it has a clear invariant subspace structure. As a result, each of these Hamiltonians is parallel-learnable and can be transformed into a direct sum of \( 2 \times 2 \) non-diagonal matrices through simple permutations. Moreover, for this particular class of Hamiltonians, the invariant subspace structure naturally emerges in the computational basis. Consequently, within the learning algorithm, we can implement a virtual permutation by preparing the appropriate encoded initial states corresponding to the defined invariant subspaces, which is effectively equivalent to permuting the Hamiltonian itself.

Lastly, the invariant subspace structure of a parallel-learnable Hamiltonian is preserved under time evolution. This leads to the following lemma concerning the form of the propagator associated with such a Hamiltonian.
\begin{lemma}
    The time-evolution operator (or propagator) of the parallel-learnable Hamiltonian $H$ is also block-diagonal with $2\times 2$ blocks. 
\end{lemma}
Therefore, the propagator of the parallel-learnable Hamiltonian also has $2^{n-1}$ different invariant subspaces. This property enables the simultaneous estimation of multiple parameters by leveraging time-evolution dynamics within different invariant subspaces in parallel. We present specific designs utilizing the propagator for parallel learning the estimators $\mathbf{\hc}$ in Problem~\ref{def:prb 1} in Sec.~\ref{subsec:methods}. But the techniques developed in this work can also be readily applied to other parallel-learnable Hamiltonians.

\subsection{Main Results}\label{sec:main results}
In this section, we present our main results, starting with the discussion on the quadratic speedup achieved by our learning algorithm for parallel-learnable Hamiltonians.
\begin{theorem}\label{thm:quadratic speedup}
    Considering the parallel-learnable full Hamiltonian 
    $
    H_i = a_iX_i + \sum_{{i,j \in\{1,\cdots,n\}} \atop {i\neq j}} c_{ij} Z_iZ_j,
$
there exists an algorithm which involves $O(n)$ number of independent experiments to learn all $O(n^2)$ interaction parameters in Eq.\eqref{eqn: all2all H}.
\end{theorem}
The key behind this quadratic speedup of this learning algorithm lies in the ability to exploit multiple parallel invariant subspaces of the full Hamiltonian simultaneously. In our algorithm, each independent experiment evolves a specific full Hamiltonian $H_i, \forall i \in \{1,\cdots,n-1\}$. Each single experiment round allows us to extract information about $O(n)$ parameters, thus to learn $O(n^2)$ parameters, only $O(n)$ independent experiments are required. However, parallelization involves more invariant subspaces, adding sampling overheads for each experimental round. Nevertheless, our algorithm remains favorable, as the total evolution time, accounting for both the number of rounds and the sampling per round, still achieves a $1/2$ reduction in the exponent of $n$ compared to~\cite{huang2023learning}. Further details on the total evolution time are provided in Supplementary Material Sec.~\ref{apx:total time}.

For each independent experiment with Hamiltonian $H_i$, we utilize the parallelization of the QSPE method \cite{dong2025qspe} within each invariant subspace to learn all parameters simultaneously. The QSPE method \footnote{The introduction to the QSPE method can be found in the Supplementary Material Sec.~\ref{sec:QSPE}.} was originally developed for two-qubit gate calibration, leveraging quantum signal transformation to separate parameters in the Fourier domain. It achieves the Heisenberg limit and the variance even scales as $1/d^4$ in the pre-asymptotic regime, while also demonstrating robustness to realistic experimental errors. Consequently, the estimators used in our learning algorithm retain the advantageous characteristics of the QSPE estimator.
\begin{theorem}[Optimal precision performance]\label{thm:heisenberg limit}
    The estimator $\hc_{ij},\forall i\neq j$ has the variance $\var(\hc_{ij}) = O(\frac{1}{Nd^4}),$ where $N$ is the number of shots and $d$ is the number of repetition cycles. Further, the computed variances saturate the Cramér–Rao lower bound.
\end{theorem}
Beyond this, the estimators are inherently robust to realistic errors. The estimators decouple the parameters in all invariant subspaces for all independent experiments, resulting in the time-dependent coherent errors on one of them will not affect the estimation accuracy of the other. Moreover, the estimators are designed from Fourier coefficients, enabling simple rescaling procedures in the algorithm to effectively mitigate the depolarization and coherent state preparation errors. As a result, our learning algorithm exhibits full robustness against SPAM errors, decoherence, and time-dependent coherent noise.

The resource efficiency and robustness of our learning algorithm make it well-suited for practical deployment on current quantum hardware. In particular, our method can be applied to infer interatomic distance errors in Rydberg atom arrays, which is an important source of systematic deviation between theoretical predictions and experimental results, as highlighted in prior studies \cite{wurtz2023aquila,dag2024distanceerror,tamura2020analysis,marcuzzi2017facilitation}. While these deviations have been effectively modeled using distance error hypotheses, a first-principles analysis has remained elusive. This is largely due to the fact that such distance errors are difficult to directly capture in experiments, constrained by limited coherence times and the presence of realistic noise. Our algorithm addresses this challenge by requiring only short-time evolutions to amplify the parameters of interest, enabled by estimators that achieve optimal precision scaling and saturate the Cramér–Rao bound. This feature makes it compatible with the coherence times achievable on existing NISQ-era Rydberg platforms. Furthermore, the algorithm’s robustness ensures reliable performance in the presence of realistic system noise. Consequently, our method enables efficient and robust estimation of multiple interatomic distances, which are essential for constructing accurate atom-distance error models and explaining the discrepancies observed in previous Rydberg-based experiments~\cite{dag2024distanceerror}.

\subsection{\label{subsec:methods}Methods}
In this section, we provide a detailed discussion of the learning algorithm for solving Problem~\ref{def:prb 1}. The techniques here are also directly applicable to any other parallel-learnable Hamiltonian. We focus on the full 
parallel-learnable Hamiltonian Eq.~\eqref{eqn:full H} with both interaction terms and the local fields to construct our learning algorithm.

\subsubsection{Two-atom Case}
First, we consider the simplest scenario for the Hamiltonian given in Eq.~\eqref{eqn:full H} with two atoms. The corresponding Hamiltonian is
\begin{equation}
    H(a,c_{12}) = aX_1 + c_{12}Z_1Z_2.
\end{equation}
We aim to estimate $\hc_{12}$ for the interaction term, while $\hat{a}$ can also be inferred via our learning algorithm as a byproduct. This Hamiltonian is native to the Rydberg atom system, where the local field is applied via Rabi frequency $\Omega$ and $c_{12}$ is the atom-atom interaction. Therefore, the learning algorithm can be directly applied to learn the atom-atom interaction and atom distance on the Rydberg device. Details on this implementation will be discussed in Sec.~\ref{sec:experiment proposal}.

Consider the Hamiltonian evolution for time $T$, then the corresponding time evolution operator can be written as 
\begin{align}
    \mathcal{U} &= \exp(-i\int_0^T Hdt)
    = \exp[-i(a^TX_1+c_{12}^TZ_1Z_2)]\label{eqn:two qubit U},
\end{align}
where $a^T:=\int_0^Tadt = aT$ and $c_{12}^T:=\int_0^Tc_{12}dt = c_{12}T$ \footnote{\label{footnote: square pulse} Here we apply the time-independent pulse to the system, the accumulation angle is just the time-integral for the square pulse.}. By simply dividing out the evolution time $T$, we can easily recover the original parameters $(\hat{a},\hc_{12})$. Taylor expanding the above Eq.~\eqref{eqn:two qubit U}, the corresponding unitary is $
    \mathcal{U} = \cos \omega I -i\sin \omega (\frac{a^T}{\omega}X_1+\frac{c_{12}^T}{\omega}Z_1Z_2)$
where $\omega = \sqrt{(a^T)^2+(c_{12}^T)^2}$. Defining the logical subspace as $\ket{0}_l = \ket{00}$ and $\ket{1}_l = \ket{10}$, which spans the invariant subspace $\mathcal{B}$. The unitary restricted to the subspace $\mathcal{B}$ is
\begin{align}\label{eqn:UB_2 atom}
    \mathcal{U}_{\mathcal{B}} = \begin{bmatrix}
         \cos\omega-i\frac{c_{12}^T}{\omega}\sin\omega & -i\frac{a^T}{\omega}\sin\omega  \\
       -i\frac{a^T}{\omega}\sin\omega&  \cos\omega+i\frac{c_{12}^T}{\omega}\sin\omega\\
    \end{bmatrix}.
\end{align}
The standard unitary calibrated using QSPE \cite{dong2025qspe} is in the form of 
\begin{equation}\label{eq:u gate}
    \mathcal{U} = \begin{bmatrix}
        \cos(\theta) e^{-i\zeta} & -i\sin(\theta)e^{i\chi} \\
        -i\sin(\theta)e^{-i\chi} & \cos(\theta)e^{i\zeta} 
    \end{bmatrix}.
\end{equation}
The following mapping can convert one to the other,
\begin{equation}\label{eqn: map}
    \begin{cases}
        \tan(\zeta) =  \frac{c_{12}^T}{\sqrt{(a^T)^2+(c_{12}^T)^2}}\tan(\sqrt{(a^T)^2+(c_{12}^T)^2})\\
        \sin^2(\theta) = (\frac{a^T}{\sqrt{(a^T)^2+(c_{12}^T)^2}})^2\sin^2(\sqrt{(a^T)^2+(c_{12}^T)^2})
    \end{cases},
\end{equation}
and $(a,c_{12}) = (a^T/T,c_{12}^T/T)$. Therefore, we can utilize the QSPE method for the $2\times2$ invariant subspace defined by $(\ket{00},\ket{10})$ for learning both parameters. The details regarding implementation are discussed in the Supplementary Material Sec.~\ref{sec:implement detials for two qubit}. The step-by-step algorithm is shown as follows, see Algorithm \ref{alg:qspe}.
\begin{algorithm}[H]
\caption{Two-atom QSPE-Learning algorithm}\label{alg:qspe}
\begin{algorithmic}
\Require Choose the hyper-parameters for the number of cycles, and sample size $(d,N)$.
\State 1: Compute the control parameter $\omega_j = \frac{j}{2d-1}\pi$ (for $j = 0,1,\cdots,2d-2$).
\State 2: Initiate a complex-valued data vector $\vec{h} \in \mathbb{C}^{2d-1}$.\\
\textbf{Main steps:}\\
\textbf{For} $j \in \{0,1,\cdots,2d-2\}$,\\
\begin{itemize}[label=$\blacktriangleright$, leftmargin=1.5em]
\item Quantum Part:
    \begin{itemize}
        \item Prepare the logical initial state $\ket{+_l}=  \frac{1}{\sqrt{2}}(\ket{00}+\ket{10})$ and $\ket{i_l} = \frac{1}{\sqrt{2}}(\ket{00}+i\ket{10})$.
        \item Define one cycle of evolution: apply the full Hamiltonian evolution $\mathcal{U}$ and subsequently apply the logical Z rotation $Z_l(\omega_i)$ for the specific angle $\omega_i$.
        \item Apply $d$ cycles of evolution to both initial states.
        \item Measure in computational basis and get $x_{j}^{(+)}$ and $x_{j}^{(i)}$ 00's out of $N$ samples for both initial states.
    \end{itemize}
\item Classical Part:
    \begin{itemize}
        \item Estimate the transition probability: $p_X= \frac{x_{j}^{(+)}}{N}$ and $p_Y= \frac{x_{j}^{(i)}}{N}$.
        \item Update $\vec{h}_j \leftarrow p_X(\omega_j)-\frac{1}{2}+i(p_Y(\omega_j)-\frac{1}{2})$.
    \end{itemize}
\end{itemize}
\textbf{end for}\\
\textbf{Post-processing}:
\State 3: Compute the Fourier coefficients $\vec{c} = FFT(\vec{h})$.
\State 4: Compute the phase difference vector $\vec{\Delta}$, where $
\vec{\Delta}= (\Delta_0, \cdots, \Delta_{d-2})^T$ and $\Delta_k = \operatorname{phase}(c_k \overline{c}_{k+1})$. 
\State 5: Compute estimates $(\hat{\theta},\hat{\zeta})$ with
\begin{equation}
    \hat{\theta} = \frac{1}{d}\sum_{k = 0}^{d-1}|c_k| \quad, \hat{\zeta} = \frac{1}{2}\frac{\vec{\mathbb{I}}^T\mathcal{D}^{-1}\vec{\Delta}}{\vec{\mathbb{I}}^T\mathcal{D}^{-1}\vec{\mathbb{I}}},
\end{equation}
where $\mathbb{I}$ is an all-one vector and $\mathcal{D}$ is the $(d-1)\times(d-1)$ discrete Laplacian matrix.
\State 6: Covert the estimated parameters $(\hat{\theta},\hat{\zeta})$ to $(\hat{a},\hat{c}_{12})$ by the mapping given in Eq.~\eqref{eqn: map}.
\end{algorithmic}
\end{algorithm}

\subsubsection{Multi-atom Case}
In this section, we extend the two-atom learning algorithm introduced earlier to construct a generic learning algorithm for the $n$-qubit Hamiltonian, as defined in Eq.~\eqref{eqn:full H}. The idea is to encode the full Hamiltonian into multiple 2×2 subspaces, each solvable with the QSPE as suggested in the previous section. The learning process can be parallelized, which allows us to estimate all coefficients induced by the same qubit in parallel within a single independent experiment round. To be more specific, consider the full Hamiltonian $H_i$, we can define its $2\times 2$ invariant subspaces as follows. Denote the computational bases of $n$-qubit Hamiltonian as $\{\ket{v_1},\ket{v_2}, \cdots, \ket{v_{2^n}}\}$, where $\ket{v_1} = \ket{00\cdots0}~\text{with}~|v_1| = 0$ and $\ket{v_{2^n}} = \ket{11\cdots1}~\text{with}~|v_{2^n}| = n.$ For each $H_i$, we define the encoding subspaces.
\begin{definition}\label{def:encoding}
For an $n$-qubit system with the full Hamiltonian including the $X_i$ term as given in Eq.~\eqref{eqn:full H}, the Hilbert space decomposes into $2^{n-1}$ invariant subspaces. We denote these subspaces by $\{\mathcal{B}_1, \cdots, \mathcal{B}_{2^{n-1}}\}_i$. Each invariant subspace $\mathcal{B}_k$ (for $k \in \{1, \cdots, 2^{n-1}\}$) is two-dimensional and corresponds to a logical subspace. The encoding for each $\mathcal{B}_k$ is defined as:
\begin{align}\label{eqn:def of B basis}
& \ket{v_m} := \ket{\bar{0}}, \quad \ket{v_n} := \ket{\bar{1}}, 
~\text{with}~v_m \oplus v_n = e_i,
\end{align}
where $m, n \in \{1, \cdots, 2^n\}$ with $m \neq n$, and $e_i$ is the bitstring with $1$ at the $i$-th position and $0$ elsewhere. Each pair $(\ket{v_m}_k, \ket{v_n}_k)$ defines one of the $2^{n-1}$ invariant subspaces $\mathcal{B}_k$.
\end{definition}
Under the above encoding, the full Hamiltonian in Eq.\eqref{eqn:full H} can be represented as a block diagonal Hermitian matrix with block size $2$. Moreover, the entries of each $2\times2$ matrix $H_i^k$ restricted to the $k$th invariant subspace can be computed by considering the projection of the $H_i$ onto each $\mathcal{B}_k$. The off-diagonal entries are introduced by the local bit-flip term $X_i$, so they are identically equal to $a_i$ for all invariant subspaces. The diagonal entries are linear combinations of all $\frac{n(n-1)}{2}$ interaction term parameters. Specifically, we can write down the analytical form of the $H_i^k$ as follows.

\begin{lemma}
Consider $H_i^k$, which is the projection of $H_i$ onto invariant subspace $\mathcal{B}_k$ spanned by $\{(\ket{v_m}_k,\ket{v_n}_k)\}$. Let $E_i$ denote the set of edges associated with vertex $i$. Specifically, the edge set is $E_i = \{(i,1),(i,2),\cdots, (i,i-1),(i,i+1),\cdots,(i,n)\}$ in the all-to-all connected system. We classify the interaction terms into two classes:\\ 
1. The different index set, denoted by $E_D$, includes all operators that act on edges connected to vertex $i$, which result in differing signs in the diagonal elements of $H_i^k$. It is defined as
\begin{equation}
E_D = \{(p,q) \mid (p,q) \in E_i\} = \{(p,q) \mid p = i\}.
\end{equation}
2. The same index set, denoted by $E_S$, includes all remaining operators that do not act on the edges connected to vertex $i$, and thus yield the same sign in the diagonal elements. It is defined as
\begin{equation}
E_S = \{(p,q) \mid (p,q) \notin E_i\} = \{(p,q) \mid p \neq i\}.
\end{equation}
Note that the second expression in both definitions of $E_D$ and $E_S$ holds specifically for the case of an all-to-all connected system.

Let $\lambda_{pq}^m$ denote the eigenvalue of the operator $Z_p Z_q$ with respect to the logical zero state $\ket{v_m}$ and set it equal to $\lambda_{pq}$, then the corresponding eigenvalue of the same operator to the logical one state $\ket{v_n}$ is given by
\begin{equation}
     \lambda_{pq}^n = \begin{cases}
         \lambda_{pq} \quad \text{when} \quad (p,q) \in E_S;\\
         -\lambda_{pq} \quad \text{when} \quad (p,q) \in E_D.
     \end{cases}
\end{equation}
The projected Hamiltonian $H_i^k$ has the following analytical matrix form:
\begin{equation}
    H_i^k = \begin{bmatrix}
        \Lambda_+ & a_i\\
        a_i & \Lambda_-
    \end{bmatrix}.
\end{equation}
where $\Lambda_+ = \sum\limits_{(p,q)\in E_S}\lambda_{pq} c_{pq} + \sum\limits_{(p,q)\in E_D}\lambda_{pq} c_{pq} $ and $\Lambda_- = \sum\limits_{(p,q)\in E_S}\lambda_{pq} c_{pq} -\sum\limits_{(p,q)\in E_D}\lambda_{pq} c_{pq} $.
\end{lemma}

The information of an analytical linear combination of our desired estimators is encoded at each invariant subspace of the Hamiltonian, and then we proceed to use the QSPE method to infer them with respect to each block. More specifically, let us denote the time-integrals of distinct entries in $H_i^k$ as follows,
    \begin{align}\label{eqn: integral}
        (A_i,B_i,C_i):&=\int_0^T(a_i,\sum\limits_{(p,q)\in E_D}\lambda_{pq} c_{pq} ,\sum\limits_{(p,q)\in E_S}\lambda_{pq} c_{pq})dt;\\
        & = (a_iT,\sum\limits_{(p,q)\in E_D}\lambda_{pq} c_{pq}T,\sum\limits_{(p,q)\in E_S}\lambda_{pq} c_{pq}T),
    \end{align} where $T$ is the evolution time \footnotemark[\value{footnote}].
    The time-evolution operator of the $k$th block is 
    \begin{equation}
        \mathcal{U}^k = \begin{bmatrix}
        \cos\omega-i\sin\omega\frac{B_i}{\omega} & -i\sin\omega\frac{A_i}{\omega}\\
        -i\sin\omega\frac{A_i}{\omega}& \cos\omega+i\sin\omega\frac{B_i}{\omega}\\
    \end{bmatrix},
    \end{equation}
    where $\omega =\sqrt{ A_i^2+B_i^2}.$
Correspondingly, the following mapping can lead to the conversion from the unitary to the one used in the QSPE method,
\begin{equation}\label{eqn: map_j}
    \begin{cases}
        \tan(\zeta) =  \frac{B_i}{\sqrt{A_i^2+B_i^2}}\tan(\sqrt{A_i^2+B_i^2})\\
        \sin^2(\theta) = (\frac{A_i}{\sqrt{A_i^2+B_i^2}})^2\sin^2(\sqrt{A_i^2+B_i^2})
    \end{cases}.
\end{equation}
For each specific block, both $A_i$ and $B_i$ can be learned. Consequently, the linear combination of these parameters is effectively learned through Eq.~\eqref{eqn: integral}. Since there are in total $|E_D|$ parameters connecting to the $i$th qubit, we require $|E_D|$ independent linear combinations of the parameters $c_{pq}$ for $(p,q) \in E_D$, obtained from $|E_D|$ distinct invariant subspaces. Specifically, we need to identify $(n-1)$ invariant subspaces whose diagonal entries differ in a linearly independent manner for the all-to-all connected Hamiltonian. We justify the existence of such subspaces and provide a method for selecting these $(n-1)$ linearly independent subspaces in the following theorem.
\begin{figure*}[t]\centering\includegraphics[width=\linewidth]{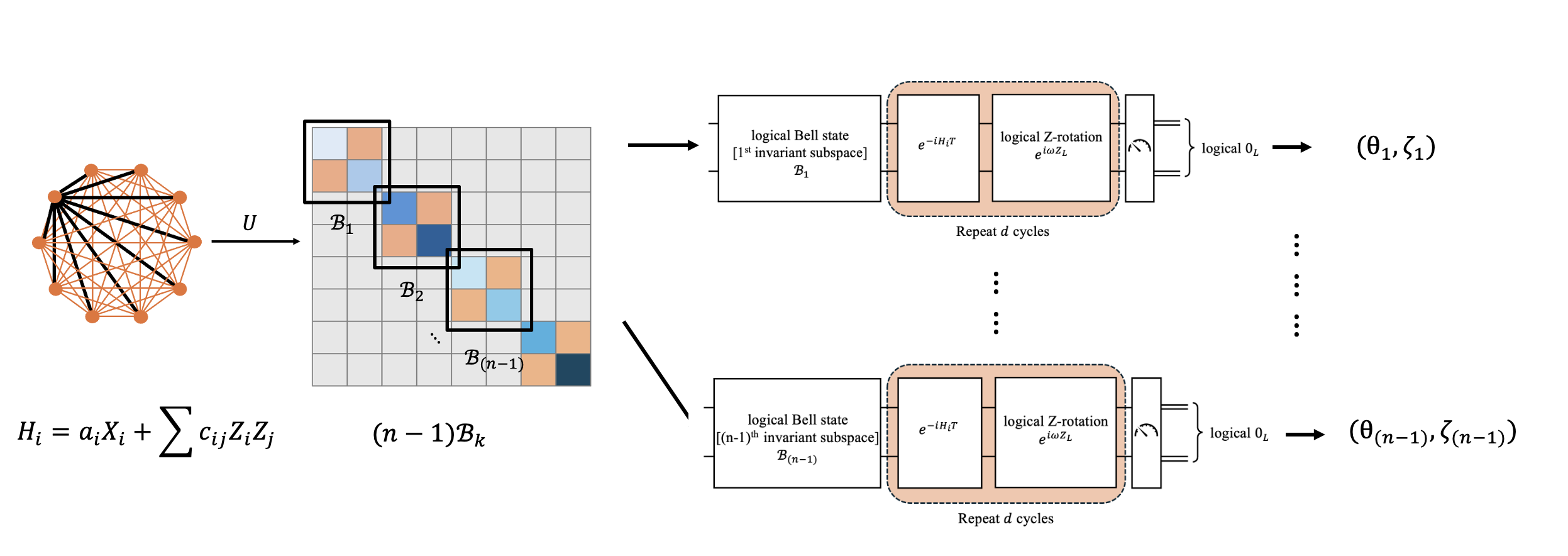}
        \caption{General framework for multi-qubit parallel learning. The $i$th full Hamiltonian $H_i$ is used to learn the $(n-1)$ parameters associated with the $i$th qubit. First, $H_i$ is mapped to a block-diagonal form with block size $2$, of which $(n-1)$ blocks are employed for estimation. Each block corresponds to an invariant subspace $\mathcal{B}_k$ with Definition \ref{def:encoding}. On each block for invariant subspace $\mathcal{B}_k$, QSPE is applied to estimate a parameter pair $(\theta_k,\zeta_k)$ encoding the Hamiltonian information via Eq.~\eqref{eqn: map_j}. The QSPEs across all invariant subspaces can be executed in parallel by evolving under the multi-qubit Hamiltonian, followed by the logical $Z_L(\omega)$ for phase accumulation. When projected onto each invariant subspace, the dynamics reduce to QSPE on a $2\times 2$ matrix.
}
        \label{fig:exp}
    \end{figure*}
    
\begin{theorem}
For the $n$-qubit full Hamiltonian $H_i$ defined in Eq.~\eqref{eqn:full H}, there exist $2^{n-1}$ distinct nonzero invariant subspaces, each associated with an encoding basis of the form $(\ket{v_m}_s, \ket{v_n}_s)$, where $s = 1, \dots, 2^{n-1}$ and $v_m \oplus v_n = e_i$. We use one of the logical zero state $\ket{v_m}$ or logical one state $\ket{v_n}$ from each pair to label the corresponding invariant subspace. Define the set
    \(
    \mathcal{S}_0 = \left\{ \ket{v_m}_{1}, \dots, \ket{v_m}_{2^{n-2}} \,\middle|\, b_i b_{i+1} = 00 \right\},
    \)
    which contains those basis states with $0$ in both the $i$-th and $(i+1)$-th positions. The complementary set, denoted by
    \(
    \mathcal{S}_1 = \left\{ \ket{v_n}_{2^{n-2}+1}, \dots, \ket{v_n}_{2^{n-1}} \,\middle|\, b_i b_{i+1} = 11 \right\},
    \)
    includes the remaining basis states where both the $i$-th and $(i+1)$-th bits are $1$.
    Then we get the following properties:
    \begin{enumerate}
        \item $|\mathcal{S}_0|+|\mathcal{S}_1| = 2^{n-1}$ representing all non-zero invariant subspaces.
        \item Define the coefficient vector $w_k: = [\lambda_{i1},\cdots,\lambda_{in}]^T$ representing the coefficient of the linear combination for the different sign term $\sum\limits_{(p,q)\in E_D}\lambda_{pq} c_{pq}$ for basis state $\ket{v_k}$. The space spanned by the coefficient vectors associated with $\mathcal{S}_0$ is the same as $\mathcal{S}_1$.
        \item Define the coefficient matrix $W:= [w_1, \cdots, w_{2^{n-2}}]$ by putting all coefficient vector for all basis in $\mathcal{S}_0$. It is a full rank matrix with size $(n-1)\times 2^{n-2}$.
        \item By selecting $(n-1)$ linearly independent columns indexed by a subset $\mathcal{C} \subset \{1, \dots, 2^{n-2}\}$, from the coefficient matrix, one can solve for the $(n-1)$ parameters. The corresponding invariant subspaces chosen are those spanned by the $(n-1)$ distinct basis vectors $\ket{v_m}_s$ with $s \in \mathcal{C}$.
    \end{enumerate}
\end{theorem}

\begin{proof}
    The first property is straightforward. Since there are only four possible combinations of the bit pair $b_i b_{i+1}$, namely $\{00, 01, 10, 11\}$, and each encoding pair for the invariant subspace satisfies $v_m \oplus v_n = e_i$, the pairing structure effectively partitions the full basis of the Hilbert space into four equal subsets. As a result, we have $|\mathcal{S}_0|= |\mathcal{S}_1| = 2^{n}/4 = 2^{n-2}$. 
    
    The second property follows from the observation that
    \[
    Z_i Z_{i+1} \ket{s} = Z_i Z_{i+1} \ket{t},
    \]
    for all $\ket{s} \in \mathcal{S}_0$ and $\ket{t} \in \mathcal{S}_1$. Moreover, since the bit values at all other positions are identical between elements of $\mathcal{S}_0$ and $\mathcal{S}_1$, applying $Z_i Z_j$ with $j \in \{1, \dots, n\},\ j \neq i, i+1$ yields the same set of eigenvalues $\lambda_{ij}$ across both sets. Therefore, the spaces spanned by coefficient vectors from $\mathcal{S}_0$ and $\mathcal{S}_1$ are the same. Therefore, it is sufficient to focus on one of the basis sets to extract all the necessary information to solve for the parameters. In the discussion of the third and fourth properties, we restrict our attention to $\mathcal{S}_0$.
    
    Consider the $Z_iZ_{i+1}$ operator, for any basis in $\mathcal{S}_0$, 
    \begin{equation}
        Z_iZ_{i+1}\ket{b_1\cdots00\cdots b_n} = \ket{b_1\cdots00\cdots b_n}.
    \end{equation}
    So $\lambda_{i,i+1} = 1$ for all basis in $\mathcal{S}_0$, i.e., the $i$th term is $1$ for all $w_k$. Each $w_k$ contains the coefficients of the different sign terms, and thus has dimension $(n-1)$ with $i$th entry equal to $1$. Besides, there are $2^{n-2}$ elements in $\mathcal{S}_0$, resulting the coefficient matrix $W$ has the size $(n-1)\times 2^{n-2}$. To prove the matrix $W$ is of full rank, we show there exists a submatrix $W_{s} \in \mathbb{R}^{(n-1)\times (n-1)}$ that is of full rank. We choose the submatrix $W_s$ as follows. 
\begin{definition}[Choice of invariant subspace and its coefficient matrix]\label{def:invariant subspace}
    The logical zero basis states for $(n-1)$ invariant subspaces are
    \begin{equation}
\left\{
\begin{aligned}
    &\ket{v_m}_{i} = \ket{00\cdots 0\cdots00}, &&\text{all-zero state} \\
    &\ket{v_m}_{k} = \ket{01\cdots 0\cdots00}, &&\text{all entries are } 0 \text{ except the $k$th}
\end{aligned}
\right.
\end{equation}
for $k \in \{1, \cdots, n\},\; k \neq i, i+1$.
\end{definition}
Applying $Z_iZ_j$ where $j \in \{1,\cdots,n\}$ and $j\neq i$ to the bases $\{\ket{v_m}_{1},\cdots,\ket{v_m}_{{n-1}}\}$, the corresponding coefficient matrix is
\begin{equation}
    W_s = \begin{bmatrix}
        -1 & 1 & 1 & \cdots & 1\\
        1 & -1 & 1 &\cdots  & 1\\
        1 & 1  & 1 & \cdots & 1\\
        \vdots &\vdots & \vdots & \vdots&\vdots\\
        1 & 1 & 1& \cdots & -1
    \end{bmatrix},
\end{equation}
where for the $k$th column ($k = 1, \cdots, n$ and $k \neq i,\, i{+}1$), the vector $w_k$ has an entry of $-1$ at the $k$th row and $1$ in all other entries, while the $i$th column consists entirely of $1$'s. To prove this submatrix $W_s$ is of full rank, we compute this determinant. Subtracting column $i$ from all other columns, the determinant of this transformed matrix is $1\times (-2)^{n-2} \neq 0$, thus the $W_s$ matrix is full-rank.

Finally, solving for $(n-1)$ different coefficients requires forming $(n-1)$ linearly independent equations. Since there are in total $(n-1)$ linearly independent columns from a total $2^{n-1}$ column matrix given property three, we can pick up any set of linearly independent ones, for example, as defined in Definition~\ref{def:invariant subspace}. The chosen invariant subspaces $\{\mathcal{B}_1,\cdots,\mathcal{B}_{(n-1)}\}$ are then spanned by encoding pairs $\{(\ket{v_m}_{1},\ket{v_n}_{1}),\cdots (\ket{v_m}_{{{n-1}}},\ket{v_n}_{{n-1}})\}.$ 
\end{proof}
The above discussion confirms that all parameters associated with the same qubit (WOLG, saying atom $i$) can be recovered in one experiment round with the full Hamiltonian $H_i$. Therefore, the interaction parameters of all-to-all connection Hamiltonian given in Eq.~\eqref{eqn: all2all H} can be recovered with $O(n)$ independent rounds of experiments. As a result, our learning algorithm can achieve a quadratic speedup in terms of the number of experiment rounds in the whole process, as stated in the following Theorem \ref{thm:quadratic speedup}. 

For each Hamiltonian term $H_i$, one can exploit $(n-1)$ invariant subspaces to extract the corresponding $(n-1)$ parameters associated with the $i$th qubit. However, when designing the learning algorithm, we can adopt an iterative approach to maximize the information gained in each round. As the iteration over $i$ progresses, the number of unknown parameters decreases to $(n-i)$ for each $H_i$. Consequently, it suffices to employ only $(n-i)$ invariant subspaces at each step by preparing initial states restricted to these subspaces and evolving the system dynamics accordingly. In this case, the coefficient matrix is updated to $W_s^i$, an $(n-i)\times (n-1)$ matrix consisting of $(n-i)$ linearly independent rows of $W_s$, corresponding precisely to the selected $(n-i)$ invariant subspaces.

The learning algorithm within each invariant subspace follows the same procedure as described in the two-atom case in Algorithm~\ref{alg:qspe}, while extending to multiple parallel invariant subspaces. We introduce modifications to both the state preparation procedure and the implementation of the logical operation $Z_L(\omega)$, as detailed below. Furthermore, in accordance with the capabilities of the quantum hardware, we develop two distinct algorithmic variants: one designed for analog–digital hybrid architectures and another for fully analog systems. In addition, the algorithm can also be executed on a fully digital quantum device by replacing the Hamiltonian evolution step in the hybrid mode with quantum gates for learning the parameters of the quantum gates. 

\paragraph{Analog-digital-hybrid learning algorithm}
In the analog–digital hybrid mode, the quantum dynamics is assumed to consist not solely of continuous-time Hamiltonian evolution but rather a combination of both continuous evolution and digitized gate operations. This flexibility enables the simultaneous implementation of the logical operator $Z_L(\omega)$ across $(n-i)$ invariant subspaces.

\begin{definition}[State preparation in hybrid setting]\label{def:initial state}
Consider the invariant subspaces spanned by 
\[
\left\{ (\ket{v_m}_{1}, \ket{v_n}_{1}), \cdots, (\ket{v_m}_{{n-i}}, \ket{v_n}_{{n-i}}) \right\},
\]
then the logical states \( \ket{+_L} \) and \( \ket{i_L} \) are defined as
\begin{equation}
\left\{
\begin{aligned}
    \ket{+_L} &= \frac{1}{\sqrt{2(n-i)}}\sum_{j = 1}^{n-i} (\ket{v_m}_{j} + \ket{v_n}_{j}) \\
    \ket{i_L} &= \frac{1}{\sqrt{2(n-i)}}\sum_{j = 1}^{n-i} (\ket{v_m}_{j} + i\ket{v_n}_{j})
\end{aligned}
\right.
\label{eq:logical_states}
\end{equation}
\end{definition}
The initial state corresponds to the superposition of the Bell states on $(n-i)$ invariant subspaces, where $i = 1,\cdots (n-1)$.

\begin{definition}[Logical $Z_L$ in hybrid setting]\label{def:logical z}
The logical $Z_L(\omega)$ is implemented by
\begin{equation}
    Z_L(\omega) = \exp(-i\omega Z_i),
\end{equation}
where $Z_i =  I \otimes Z \otimes \cdots \otimes I$ acting $Z$ on $i$th qubit, which performed as $Z(\omega)$ at each invariant subspaces of the full Hamiltonian. 
\end{definition}
In this setting, the logical operator $Z_L(\omega)$ is realized as a digital gate, during which the analog Hamiltonian evolution is suspended. Based on these modifications, we present the Analog-digital-hybrid parallel learning algorithm \ref{alg:qspe-parallel}. for Problem \ref{def:prb 1}.
\begin{algorithm}[H]
\caption{Analog-digital-hybrid parallel learning algorithm}\label{alg:qspe-parallel}
\begin{algorithmic}
\For{$i \gets 1$ \textbf{to} $n-1$}\\  
    \quad 1: Prepare $H_i = a_iX_i + \sum_{{i,j \in\{1,\cdots,n\}} \atop {i\neq j}} c_{ij} Z_iZ_j$, the $i$th full Hamiltonian which is parallel-learnable.\\
    \quad 2: Prepare $UH_iU^{\dagger}$.
    \Comment{turning $H$ into direct sum form by applying realizable transformation $U$, $U = I$ in Problem \ref{def:prb 1}}\\
    \quad 3: Choose $(n-i)$ invariant subspace of $H_i$. \Comment{Infer $(n-i)$ parameters, $\vec{c}_i = \{c_{i(i+1)},\cdots,c_{in}\}$.}\\
    \quad \textbf{for} $j \leftarrow 0$ \textbf{to} $2d-2$ \textbf{do}
\begin{itemize}[label=$\blacktriangleright$, leftmargin=3em]
    \item Quantum Part:
    \begin{itemize}
        \item Prepare the initial state $\ket{+_L}$ and $\ket{i_L}$ according to Definition \ref{def:initial state}.
        \item Define one cycle of evolution: apply the full Hamiltonian evolution for time $T$ and subsequently apply the logical Z rotation gate $Z_L(\omega_j)$ according to Definition \ref{def:logical z}.
        \item Apply $d$ cycles of evolution on both initial states.
        \item Measure in computational basis and collect the outcome bitstrings. Count the bitstring number corresponds to the logical zero state $\{\ket{v_m}_k\}_{k = 1,\cdots,n-i}$ for all $k$ invariant subspaces and denote the number as $x_{\omega_j,k}^{(+)}$ and $x_{\omega_j,k}^{(i)}$ for both initial states.
    \end{itemize}
    \item Classical Part:
    
    \quad \textbf{for} $k \leftarrow 1$ \textbf{to} $n-i$ \textbf{do}
    \begin{itemize}
        \item Estimate the transition probability: $\hat{p}_{X,k}= \frac{x_{\omega_j,k}^{(+)}}{N}$ and $\hat{p}_{Y,k}= \frac{x_{\omega_j,k}^{(i)}}{N}$.
        \item Update the reconstruction function with rescaled probability $\vec{h}_j^k \leftarrow [(n-i)\hat{p}_{X,k}-\frac{1}{2}]+i[(n-i)\hat{p}_{Y,k}-\frac{1}{2}]$.
        \item Use the same post-processing method in Algorithm \ref{alg:qspe} to get $\{\theta_k,\zeta_k\}$ for each subspace.
        \item Solve for $B_k$, the time integral of the linear combination of the elements in $\vec{c}_{i} $ via $\{\theta_k,\zeta_k\}$, according to Eq.~\eqref{eqn: map_j}.
    \end{itemize}
\end{itemize}
  \quad  \textbf{end for}\\
    \quad 4: Solve for $(n-i)$ unknown parameters $\vec{c}_i$ with the linear equation
    \begin{equation}\label{eq:linear sys}
        [B_1,\cdots,B_{n-i}]^T = W_s^i[\hc_{i1}T,\cdots,\hc_{in}T]^T,
    \end{equation}
    where $[\hc_{i1},\cdots,\hc_{i(i-1)}]$ is estimated in previous iterations.
        \EndFor
\end{algorithmic}
\end{algorithm}
From the above parallel-learning algorithm, we notice that the total number of parameters can be learned through $(n-1)$ independent experiments is
\(
    (n-1)+(n-2)+\cdots +1 = n(n-1)/2.
\)
Thus, the algorithm learns the $O(n^2)$ coefficients in $O(n)$ rounds, leading to a quadratic speedup compared to \cite{huang2023learning}. We also emphasize that the \( (n-1) \) full Hamiltonians used in the algorithm include all interactions, enabling in-situ learning of all interaction parameters simultaneously. 

\paragraph{Fully analog learning algorithm}
The learning algorithm can also be implemented in a fully analog setting, that is, the quantum dynamics are generated by continuous-time Hamiltonian evolution without any discrete quantum gates. This setting is particularly relevant to the capabilities of Rydberg-atom quantum simulators, where the inter-atomic interactions are determined by the spatial separation between atoms and remain always-on unless the atoms are physically rearranged during the evolution. In this case, the challenge is to implement the logical $Z_L(\omega)$ rotation in the presence of the always-on $ZZ$ interactions. Assuming limited controllability of the device, with only tunable single-qubit $Z$ terms available, it is not possible to realize the logical operation $Z_L$ simultaneously across all $(n-1)$ invariant subspaces. Instead, we implement the logical $Z_L^k$ rotation targeting one specific invariant subspace $\mathcal{B}_k$ at a time.

\begin{definition}[Logical $Z_L^k$ in analog setting]\label{def:logical-analog}
For the $k$th invariant subspace, the logical $Z_L(\omega)$ can be implemented with the Hamiltonian 
\begin{equation}
    H_{Z_L^k} = b_kZ_k+\sum_{{i,j \in\{1,\cdots,n\}} \atop {i\neq j}} c_{ij} Z_iZ_j,\notag
\end{equation}
where $b_k = \omega - \Lambda_+$. The logical operation is $Z_L^k(\omega) = \exp(-i H_{Z_L^k}t)$. In comparison with the full Hamiltonian given in Eq.~\eqref{eqn:full H}, this Hamiltonian is realized by suspending the single-qubit $X$ terms and activating the single-qubit $Z$ terms, while retaining all two-local $ZZ$ interactions.
\end{definition}
The implementation of the logical operator $Z_L^k$ is fully analog, thereby enabling the learning algorithms to be executed with minimal requirements on the quantum device. Moreover, since all Hamiltonian interactions remain always-on, this setting naturally supports in-situ learning of all Hamiltonian terms. However, this trade-off prevents the realization of the logical gate $Z_L$ across all invariant subspaces simultaneously, resulting in an $O(n^2)$ overhead in the number of experiments. Fortunately, in this mode, the initial states of the algorithm can be prepared specifically for the $k$th invariant subspace, rather than employing a superposition of $O(n)$ Bell pairs, thereby reducing the sampling overhead. To be more concrete, the initial state for the $k$th invariant subspace is defined as
\begin{definition}[State preparation in analog setting]\label{def:state prep-analog}
For the $k$th invariant subspace, which is spanned by the basis $\{\ket{v_m}_{k},\ket{v_n}_{k}\}$, then the logical states $\ket{+_L^k}$ and $\ket{i_L^k}$ are defined as
\begin{equation}
\left\{
\begin{aligned}
    \ket{+_L^k} &= \frac{1}{\sqrt{2}} (\ket{v_m}_{k} + \ket{v_n}_{k}) \\
    \ket{i_L^k} &= \frac{1}{\sqrt{2}}(\ket{v_m}_{k} + i\ket{v_n}_{k})
\end{aligned}
\right.
\label{eq:logical_states_analog}
\end{equation}
\end{definition}
Compared to Eq.\eqref{eq:logical_states}, the normalization factor is $(n-i)$ times larger, which will reduce the estimation variance by a factor of $n$. As a result, the total evolution time required by the fully analog learning algorithm is not expected to be significantly longer than that of the analog–digital hybrid mode. More details on the total evolution time scaling can be found in Supplementary Material Sec.~\ref{apx:total time}. 
To conclude, we present the fully analog in-situ learning algorithm as follows in Algorithm \ref{alg:qspe-analog}.
\begin{algorithm}[H]
\caption{Fully analog in-situ learning algorithm}\label{alg:qspe-analog}
\begin{algorithmic}
\For{$i \gets 1$ \textbf{to} $n-1$}\\  
    \quad 1: Prepare $H_i = a_iX_i + \sum_{{i,j \in\{1,\cdots,n\}} \atop {i\neq j}} c_{ij} Z_iZ_j$, the $i$th full Hamiltonian which is parallel-learnable.\\
    \quad 2: Prepare $UH_iU^{\dagger}$.
    \Comment{turning $H$ into direct sum form by applying realizable transformation $U$, $U = I$ in Problem \ref{def:prb 1}}\\
    \quad 3: Choose $(n-i)$ invariant subspace of $H_i$.\Comment{Infer $(n-i)$ parameters, $\vec{c}_i = \{c_{i(i+1)},\cdots,c_{in}\}$}\\
    \quad \textbf{for} $k \leftarrow 1$ \textbf{to} $n-i$ \textbf{do}\\
    \quad \quad \textbf{for} $j \leftarrow 0$ \textbf{to} $2d-2$ \textbf{do}
\begin{itemize}[label=$\blacktriangleright$, leftmargin=4em]
    \item Quantum Part:
    \begin{itemize}
        \item Prepare the initial state $\ket{+_L}$ and $\ket{i_L}$ according to Definition \ref{def:state prep-analog}.
        \item Define one cycle of evolution: apply the full Hamiltonian evolution for time $T$ and subsequently apply the logical Z rotation gate $Z_L(\omega_j)$ according to Definition \ref{def:logical-analog}.
        \item Apply $d$ cycles of evolution on both initial states.
        \item Measure in computational basis and collect the outcome bitstrings. Count the bitstring number corresponds to the logical zero state $\{\ket{v_m}_k\}$ and denote the number as $x_{\omega_j}^{(+)}$ and $x_{\omega_j}^{(i)}$ for both initial states.
    \end{itemize}
    \item Classical Part:
    \begin{itemize}
        \item Estimate the transition probability: $\hat{p}_{X,k}= \frac{x_{\omega_j,k}^{(+)}}{N}$ and $\hat{p}_{Y,k}= \frac{x_{\omega_j,k}^{(i)}}{N}$.
        \item Update the reconstruction function with rescaled probability $\vec{h}_j^k \leftarrow [\hat{p}_{X,k}-\frac{1}{2}]+i[\hat{p}_{Y,k}-\frac{1}{2}]$.
        \item Use the same post-processing method in Algorithm \ref{alg:qspe} to get $\{\theta_k,\zeta_k\}$ .
        \item Solve for $B_k$, the time-integral of the linear combination of the elements in $\vec{c}_{i} $ via $\{\theta_k,\zeta_k\}$, according to Eq.~\eqref{eqn: map_j}.
    \end{itemize}
\end{itemize}
  \quad\quad   \textbf{end for}\\
  \quad  \textbf{end for}\\
    \quad 4: Solve for $(n-i)$ unknown parameters $\vec{c}_i$ with the linear equation \eqref{eq:linear sys}, where $[\hc_{i1},\cdots,\hc_{i(i-1)}]$ is estimated in previous iterations.
        \EndFor
\end{algorithmic}
\end{algorithm}
To validate the correctness of our learning algorithm in the multi-atom setting, we simulate its performance on all-to-all connected Hamiltonians of varying sizes. In Figure.~\ref{fig:fig1}, we learn the all-to-all connected Hamiltonian for $5$ qubit with Algorithm \ref{alg:qspe-parallel}. With $d=10$ cycle and $N_{\text{shot}} = 10^4$ shots, the learning algorithm successfully infers all the Hamiltonian parameters. While in Figure.~\ref{fig:fig2}, we further demonstrate the algorithm's stable performance across different system sizes. The black bold link in Figure.~\ref{fig:fig2} indicates which interaction parameter is presented, but the performance for all of the parameters is identical. Since only one of the interactions is of interest here, we utilize Algorithm \ref{alg:qspe-analog} while choosing two specific invariant subspaces to infer that particular parameter. With $N_{\text{shot}} = 10^5$ shots and $N_{\text{bootstrap}} = 10^3$ bootstraps, which is used to estimate the variance of the estimator, we notice the variance of the estimators scales as $1/d^{4}$ as suggested in Theorem \ref{thm:heisenberg limit}.

\begin{figure}[htbp]
  \centering
  \begin{subfigure}[b]{0.5\textwidth}
    \includegraphics[width=9cm, height=4cm]{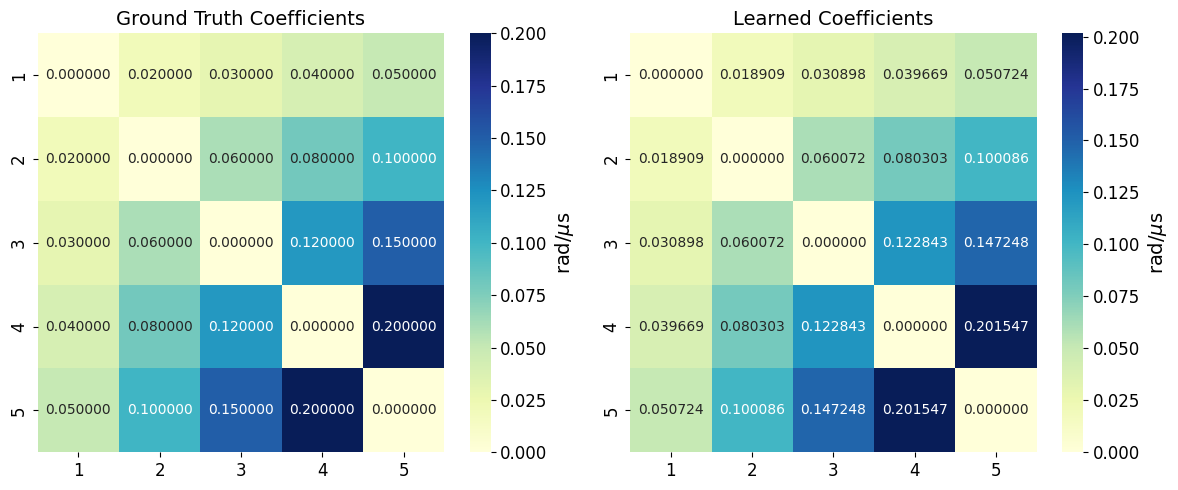}
    \caption{Application of the learning algorithm for all-to-all connected $5$ qubit system. The $(i,j)$th entry in the figure is the coefficient $c_{ij}$ for the term $Z_iZ_j$. The simulation utilizes $d = 10$ cycles and the number of samples $N_{shot} = 10^4$.}
    \label{fig:fig1}
  \end{subfigure}
  \hfill
  \begin{subfigure}[b]{0.5\textwidth}
    \includegraphics[width=8cm, height=6cm]{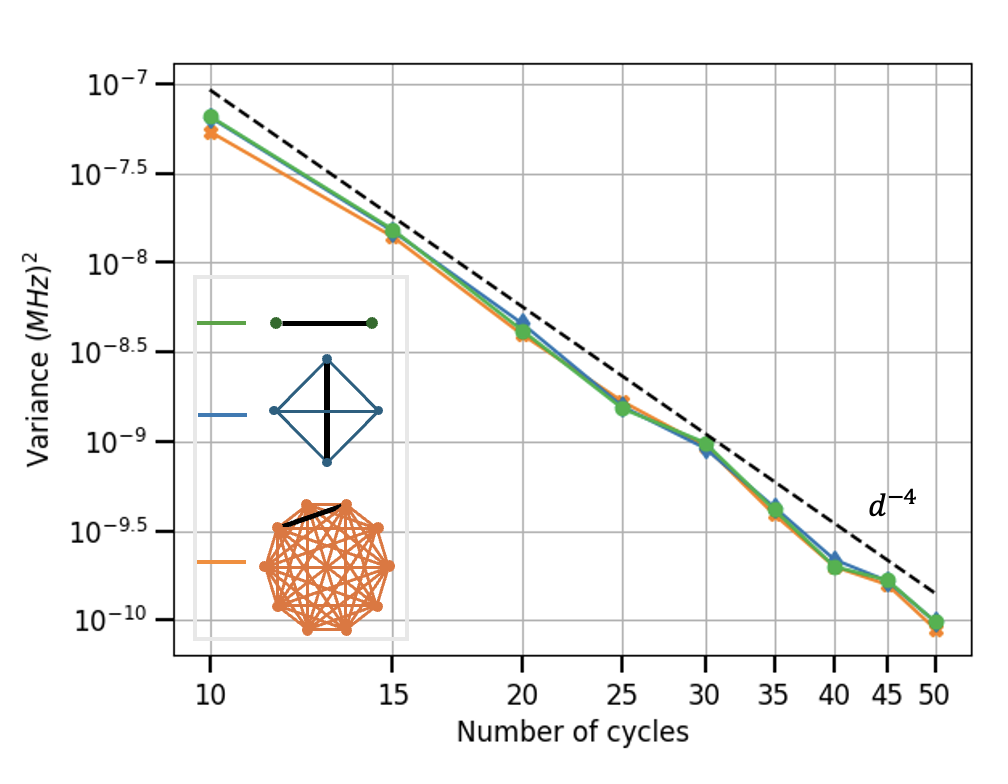}
    \caption{Scalability test on the learning algorithm. The simulation is done for $2$-qubit (green), $4$-qubit(blue) and the $10$-qubit system (orange) with the number of samples $N_{\text{shot}} = 10^5$ and the number of bootstraps $N_{\text{bootstrap}}  = 10^3$.}
    \label{fig:fig2}
  \end{subfigure}
  \caption{Performance of the algorithm for multi-atom case.}
  \label{fig:combined}
\end{figure}

\subsection{\label{subsec:optimality discussion}Optimality discussion}
In this section, we establish that the estimators in our learning algorithm achieve optimal precision under classical post-processing, with the resulting variance saturating the Cramér–Rao bound.

For both learning algorithms, the Hamiltonian parameters $\mathbf{c}$ are learned via estimators $\mathbf{\hc}$ through $(n-1)$ different full Hamiltonians, each specifying by the single $X_i$ term. Both algorithms share the same inference procedure for learning $\{\theta_k,\zeta_k\}$ on each $k$th invariant subspace and recovering the Hamiltonian parameters through an identical set of linear equations. The key distinction is that, in the analog–digital hybrid mode, all invariant subspaces evolve simultaneously using a superposition of the corresponding logical Bell pairs, whereas in the fully analog mode, only a single invariant subspace evolves at a time. Consequently, the sole difference in the inference analysis is that the transition probability for Algorithm~\ref{alg:qspe-parallel} on each invariant subspace is reduced by a factor of $O(1/n)$ relative to Algorithm~\ref{alg:qspe-analog}. Under the same number of measurement shots, this implies that the variances of the estimates $(\hat{\theta}_k,\hat{\zeta}_k)$ obtained from Algorithm~\ref{alg:qspe-parallel} are larger by a factor of $O(n)$ compared to those from Algorithm~\ref{alg:qspe-analog}. Since for each invariant subspace, we utilize the QSPE method \cite{dong2025qspe} for learning $(\hat{\theta}_k,\hat{\zeta}_k)$, the same variance computation applies to Algorithm~\ref{alg:qspe-analog} for each  $2 \times 2$ invariant subspace. Consider the accumulated angle $\theta_k$ for each subspace, which satisfies $d \theta_k \ll 1$. Here, $d$ denotes the number of repetition cycles, and $\theta_k$ is the angle resolved by QSPE in the $k$th invariant subspace, computable via Eq.~\eqref{eqn: integral} and Eq.~\eqref{eqn: map_j}. Since $\theta_k$ depends on the time integral of the Hamiltonian coefficients, we can always select an appropriate evolution time $T$ to ensure this condition is met. With sample size $N$, the variances computed for the estimated pair $(\hat{\theta}_k,\hat{\zeta}_k)$ in the $k$th invariant subspace are 
\begin{equation}\label{eq:var_analog}
    \var(\hat{\theta}_k) \approx \frac{1}{8Nd^2}~,\var(\hat{\zeta}_k)\approx \frac{3}{8Nd^4(\theta_k)^2}.
\end{equation}
For Algorithm \ref{alg:qspe-parallel}, we recompute the variance of $(\tilde{\theta}_k,\tilde{\zeta}_k)$ in the $k$th invariant subspace in Supplementary Material Sec.\ref{apx:var}, which are
\begin{equation}\label{eq:var_parallel}
    \var(\tilde{\theta}_k) \approx \frac{n}{4Nd^2}~,\var(\tilde{\zeta}_k)\approx \frac{3n}{4Nd^4(\theta_k)^2}.
\end{equation}

Given the above calculation, we proceed to compute the variance for our desired estimators $
\mathbf{\hc}$ for both algorithms. Without loss of generality, we first focus on the full Hamiltonian $H_1 = a_1X_1+\sum c_{ij}Z_iZ_j$, the first learned Hamiltonian in the whole learning procedure. Let us denote the $(n-1)$ parameters associated with the first atom as $\mathbf{c_1} = (c_{12},c_{13},\cdots, c_{1n})$. We utilize $(n-1)$ invariant subspaces to solve for the estimator $\mathbf{\hat{c}_1} = (\hc_{12},\hc_{13},\cdots, \hc_{1n})$. While for each $k$th subspace, we utilize the QSPE method to solve for $(\theta_1^k,\zeta_1^k)$. Here we use the superscript $k$ to represent the $k$th invariant subspace and the subscript $1$ to represent the $H_1$. After which, we follow the mapping \begin{equation}\label{eqn: invariant subspace relation}
    \begin{cases}
        \tan(\zeta_1^k) =  \frac{m_1^k}{\sqrt{
        (a_1^T)^2+{m_1^k}^2}}\tan(\sqrt{(a_1^T)^2+{m_1^k}^2});\\
        \sin^2(\theta_1^k) = (\frac{a_1^T}{\sqrt{(a_1^T)^2+{m_1^k}^2}})^2\sin^2(\sqrt{(a_1^T)^2+{m_1^k}^2}),
    \end{cases}
\end{equation}
to estimate the parameter set $(\hat{a}_1^T,\hat{m}_1^k)$, where $\hat{m}_1^k =\hat{B}_i^k :=  \sum_{1,j}\lambda_{1j}^k \hc_{1j}T$ and $\lambda_{1j}^k$ is the coefficient for applying $Z_1Z_j$ towards the basis $\ket{v_m}_{k}$. Given the variance for both algorithms in Eq.~\eqref{eq:var_analog} and \eqref{eq:var_parallel}, the performance of the estimator pair $(\hat{a}_1^T,\hat{m}_1^k)$ for each invariant subspace can be derived with the relation given in  Eq.~\eqref{eqn: invariant subspace relation}. Since the estimator pair $(\hat{a}_1^T,\hat{m}_1^k)$ accounts for the time-integral of the Hamiltonian parameters, we can utilize a proper evolution time $T$ for the Hamiltonian, forcing $p =\sqrt{(a_1^T)^2+{m_1^k}^2}$ to be small. Under the small-angle approximation ($\sin p \approx p$ and $\tan p \approx p$), the variances for the pair $(\hat{a}_1^T,\hat{m}_1^k)$ are approximately
\begin{align}
    \mathrm{Var}(\hat{a}_1^T) &\approx \mathrm{Var}(\hat{\theta}_1^k) \approx \frac{1}{8 N d^2},\\
    \mathrm{Var}(\hat{m}_1^k) &\approx \mathrm{Var}(\hat{\zeta}_1^k) \approx \frac{3}{8 N d^4 a_1^2},
    \label{eqn:var_of_m1_analog}
\end{align}
for each of the $(n-1)$ invariant subspaces, assuming Algorithm~\ref{alg:qspe-analog}.
 The same analysis also applies to Algorithm \ref{alg:qspe-parallel}, that is, 
\begin{align}
    &\text{Var}(\tilde{a}_1^T) \approx \text{Var}(\tilde{\theta}_1^k) \approx \frac{n}{4Nd^2};\\
    &\text{Var}(\tilde{m}_1^k) \approx \text{Var}(\tilde{\zeta}_1^k) \approx \frac{3n}{4Nd^4(a_1)^2}\label{eqn:parallel var of m_1},
\end{align}
for all $(n-1)$ invariant subspaces. Lastly, we estimate our desired quantities $\hc_{1i}$'s by solving the linear system comprising $\hm_1^k$'s. Therefore, to further provide the variance for the Hamiltonian parameters, we should consider the linear relation between $c_{1i}$'s and $m_1^k$'s, i.e.
\begin{equation}\label{eqn:linear}
    \Lambda \hat{\mathbf{c}}_1 = \hat{\mathbf{m}}_1,
\end{equation}
where $\hat{\mathbf{m}}_1 := (\hat{m}_1^1,\hat{m}_1^2,\cdots,\hat{m}_1^{n-1}).$ The coefficient matrix $\Lambda$, with all entries equal to $\pm T$, encodes the linear combination relationships. To ensure that it provides both necessary and sufficient conditions for solving all $ \hat{\mathbf{m}}_1$ terms, $\Lambda$ must be of full rank. The choice of \( \Lambda \) is determined by the selection of \( (n - 1) \) invariant subspaces, one example of which is discussed in Definition~\ref{def:invariant subspace}. Based on the linear relation, the covariance matrix of $\hat{\mathbf{c}}_1$ is 
\begin{equation}
    \Sigma_{\hat{\mathbf{c}}_1}= \Lambda^{-1}\Sigma_{\hat{\mathbf{m}}_1}\Lambda^{-T},
\end{equation}
where $\Sigma_{\hat{\mathbf{m}}_1} = \text{diag}(\text{Var}(\hm_1^1),\text{Var}(\hm_1^2),\cdots,\text{Var}(\hm_1^{n-1})$. Therefore, the variance of $\hc_{1i}$ is the $(i,i)$th entry of the covariance matrix $ \Sigma_{\mathbf{\hc_1}}$, i.e.
\begin{equation}
    \text{Var}(\hc_{1i}) = \sum_{j = 1}^{n-1}(\Lambda^{-1}[i,j])^2 \text{Var}(\hm_1^j).
\end{equation}
For Algorithm \ref{alg:qspe-analog}, we plug in Eq.~\eqref{eqn:var_of_m1_analog}. The variance of $\hc_{1i}$ terms is in the form of
\begin{equation}
    \text{Var}(\hc_{1i}) \approx  \frac{3}{8Nd^4a_1^2}  \sum_{j = 1}^{n-1}(\Lambda^{-1}[i,j])^2.
\end{equation}
This leads to a general bound on the variance of $\mathbf{\hc_1}$ as the average entry of $\Lambda^{-1}$ generally scales as $O(1)$. Therefore,
\begin{equation}
    \text{Var}(\hc_{1i}) \approx  O(\frac{3}{8Nd^4a_1^2}),
\end{equation}
where $N$ is the number of shots, $d$ is the number of repetition cycles. Similarly, for Algorithm \ref{alg:qspe-parallel}, Eq.~\eqref{eqn:parallel var of m_1} is utilized and leads to 
\begin{equation}
    \text{Var}(\tilde{c}_{1i}) \approx  O(\frac{3n}{4Nd^4a_1^2}).
\end{equation}
Further, the whole learning algorithm involves $(n-1)$ different full Hamiltonian $H_i$ for which the above analysis can be directly applied to the estimator $\mathbf{\hc_i}=(\hc_{i(i+1)},\cdots,\hc_{in}) $ corresponds to each $H_i$. In general, the variance of the estimator \( \mathbf{\hat{c}} \) for both algorithms can be summarized as
\begin{align}\label{eqn:var_c_analog}
     \text{Var}(\hc_{ij}) &\approx  O\left(\frac{3}{8 N d^4 a_i^2}\right),~~~\text{for Algorithm 3} \\ \label{eqn:var_c_parallel}
     \text{Var}(\hc_{ij}) &\approx  O\left(\frac{3n}{4 N d^4 a_i^2}\right),~~~\text{for Algorithm 2} 
\end{align}
where $a_i$ is the coefficient of the single $X_i$ term for $H_i$ and $i\in \{1,\cdots, (n-1)\}, j\in \{(i+1),\cdots,n\}$.

Based on the computed variance, we show that the Heisenberg limit performance of the estimators is attained in our learning algorithm in both scenarios. The classical resource used in the learning algorithm is $N\times 2(2d-1)$, and the quantum resource used is $d$ for each parameter. The Heisenberg limit suggests the variance scales as
\begin{align}
    &\Omega \left(\frac{1}{(\text{Classical resource}\times\text{Quantum resource}^2)}\right)\\
    &=\frac{1}{(N\times 2(2d-1))\times(d)^2} = \Omega(\frac{1}{Nd^3}).
\end{align}
Our estimators, with variances given in Eq.~\eqref{eqn:var_c_analog} and Eq.~\eqref{eqn:var_c_parallel}, achieve the Heisenberg limit, exhibiting a faster variance reduction that scales as $O(1/d^4)$, similar to the performance achieved by the QSPE method \cite{dong2025qspe}.

Moreover, we further demonstrate the optimality of the estimators with respect to classical post-processing by showing that the calculated variance attains the Cramér–Rao lower bound to complete the proof of the Theorem \ref{thm:heisenberg limit}. With the same setting in the discussion of the variance computation, we assume the full Hamiltonian $H_i$ for estimating $(n-i)$ Hamiltonian parameters via $(n-i)$ invariant subspaces. Each $k$th invariant subspace solves a pair $(\theta_i^k,\zeta_i^k)$ which can be used to infer $(a_i,m_i^k)$ based on the mapping given in Eq.~\eqref{eqn: invariant subspace relation}. Similarly, for Algorithm \ref{alg:qspe-analog}, the optimal variances of the estimator pair $(\hat{\theta}_i^k,\hat{\zeta}_i^k)$ are computed from the Cramér–Rao lower bound discussed in \cite{dong2025qspe}, that is,
\begin{align}\label{eq:optvar_analog}
     \text{Var}(\hat{\theta_i^k})_{opt} \approx \frac{1}{8Nd^2}~, \text{Var}(\hat{\zeta_i^k})_{opt}  \approx \frac{3}{8Nd^4(\theta_i^k)^2}.
\end{align}
For Algorithm \ref{alg:qspe-parallel}, we recompute the optimal variance of $(\tilde{\theta}_i^k,\tilde{\zeta}_i^k)$ in the $k$th invariant subspace in Supplementary Material Sec.\ref{apx:CR bound}, which are
\begin{equation}\label{eq:optvar_parallel}
    \var(\tilde{\theta}_i^k)_{opt} \approx \frac{n}{4Nd^2}~,\var(\tilde{\zeta}_i^k)_{opt}\approx \frac{3n}{4Nd^4(\theta_1^k)^2}.
\end{equation}
Based on the same analysis that we developed for variance computation for estimators, the optimal variance for the Hamiltonian parameter estimators can be derived based on Eq.~\eqref{eq:optvar_analog} and Eq.~\eqref{eq:optvar_parallel}. To conclude, the optimal variance of the estimator \( \mathbf{\hat{c}} \) for both algorithms can be summarized as
\begin{align}\label{eqn:var_c_analog_opt}
     \text{Var}(\hc_{ij})_{opt} &\approx  O\left(\frac{3}{8 N d^4 a_i^2}\right),~~~\text{for Algorithm 3} \\ \label{eqn:var_c_parallel_opt}
     \text{Var}(\hc_{ij})_{opt} &\approx  O\left(\frac{3n}{4 N d^4 a_i^2}\right),~~~\text{for Algorithm 2} 
\end{align}
where $a_i$ is the coefficient of the single $X_i$ term for $H_i$ and $i\in \{1,\cdots, (n-1)\}, j\in \{(i+1),\cdots,n\}$. It is straightforward to verify that the optimal variance given by the Cramér-Rao lower bound agrees with the computed variance, indicating that the classical post-processing component of our learning algorithms is optimal.

\section{\label{sec:experiment deployment} Experiment deployment}
In this section, we consider the realistic scenarios for deploying the learning algorithms on NISQ devices, taking into account various noise sources commonly present in current hardware and providing a detailed experimental proposal for NISQ demonstration of our algorithm.

\subsubsection{Robustness against realistic error}\label{subsec:robustness}
\paragraph{Decoherence}
The unavoidable disturbance from the environment influences the quantum dynamics, thus might lead to the wrong estimation using our algorithms. Here we model the decoherence effect by the depolarizing channel resulting in the mismatch of the probability from the quantum part of Algorithm \ref{alg:qspe}, i.e.
\begin{equation}
    p_{X(Y)}' = \alpha p_{X(Y)}+ \frac{1-\alpha}{4},
\end{equation}
where $\alpha \in [0,1]$ quantifies the fidelity of the quantum evolution. Here, we analyze the performance of one specific invariant subspace, but due to the parallelization structure, the analysis can be generalized to the whole process. The probabilities coming from the noisy dynamics suffer a rescaled factor and a constant shift, thus the corresponding reconstructed function $h$ used in Algorithm \ref{alg:qspe} also changes accordingly, i.e. $h' = \alpha h-\frac{1-\alpha}{4}(1+i).$ The Fourier coefficients of $h$ are 
expected to have an additional rescaling factor $\alpha$, and the constant shift only contributes to the change of zeroth Fourier coefficients, namely,
\begin{align}
    &|c_0'| = |\alpha c_0 -\frac{1-\alpha}{4}(1+i)|\approx \alpha \theta+\frac{1-\alpha}{2\sqrt{2}}, \\
    &|c_k| \approx \alpha \theta, ~ \forall k = 1, \cdots, d-1.
\end{align}
Therefore, utilizing the relation between the zeroth order and the remaining Fourier coefficients \cite{dong2025qspe}, the estimator in Algorithm \ref{alg:qspe} can be redefined as follows:
\begin{equation}
    \hat{\alpha} = 1-2\sqrt{2}(|c_0|-\frac{1}{d-1}\sum_{k=1}^{d-1}|c_k|), \quad \hat{\theta} = \frac{1}{\hat{\alpha}}\times \frac{1}{d-1}\sum_{k=1}^{d-1}|c_k|.
\end{equation}
The rescaling estimator mitigates decoherence errors in the algorithm by leveraging the fact that the Fourier coefficients encode information about the estimated fidelity.

\paragraph{Coherent state preparation error}
The state preparation error is also dominant in quantum systems. Especially, in our state preparation step of Algorithm \ref{alg:qspe}, we need to accurately applying $-X(\frac{\pi}{4})$ and $Y(\frac{\pi}{4})$ pulse towards the first atom (See Supplementary Material Sec.~\ref{sec:implement detials for two qubit} for details), where miscalibration of the pulses leads to imperfect initial states and thus affects the estimation process. Here, we test the robustness of the algorithm against state preparation error in the following setting. Instead of the perfect pulse in Eq. \eqref{eq:plus_state_prep} \& \eqref{eq:i_state_prep}, we assume a coherent error happening in the pulse, which leads to over-rotation, i.e.
\begin{align}\label{eq:coherent state prep error}
      e^{-i(\theta+\alpha) X_1\otimes I}\ket{00} &= \cos(\theta+\alpha)\ket{00}-i\sin(\theta+\alpha)\ket{10}\notag\\ 
    e^{-i(\theta+\alpha)Y_1\otimes I}\ket{00} &= \cos(\theta+\alpha)\ket{00}+\sin(\theta+\alpha)\ket{10};   
\end{align}
where $\alpha$ denotes the error rate. Here we provide a bound on the difference of the estimation with and without coherent state preparation error using the original estimator and the scaled estimator.
\begin{theorem}\label{thm:coherent state prep error}
    Consider the coherent error Eq.~\eqref{eq:coherent state prep error} with error rate $\alpha$ in the state preparation step of Algorithm \ref{alg:qspe}, the difference between the noisy estimation $\hat{\theta'} = \frac{1}{d}\sum_{k = 0}^{d-1}|c_k'|$ and noiseless estimation $\hat{\theta} = \frac{1}{d}\sum_{k = 0}^{d-1}|c_k|$ can be bounded by
    \begin{equation}
      D =  |\hat{\theta}'-\hat{\theta}|\leq \sqrt{2}(d+1)^2\sin2\alpha\sin^2\theta+ (\cos2\alpha - 1)\theta,
    \end{equation}
    which scales with $O(\theta)$.
    The difference can be improved by introducing a scaled estimator $ \hat{\theta}_s' = \frac{1}{\cos2\alpha}\times \frac{1}{d}\sum_{k = 0}^{d-1}|c_k'|.$ The corresponding bound then reduces to 
    \begin{equation}
         D_s =  |\hat{\theta}_s'-\hat{\theta}|\leq  \sqrt{2}(d+1)^2\tan2\alpha \sin^2\theta,
    \end{equation}
    where scales only with $O(\theta^2)$.
\end{theorem}
Unlike the simple scaling and constant shift introduced by the depolarizing channel to the function $h$, coherent state preparation errors contribute both a rescaling factor and a shift that depend on the true parameter value $\theta$. However, by applying the same rescaling technique used to address decoherence, we reduce the error from first order to second order. As a result, when $\theta$ is sufficiently small, the induced error becomes negligible. The detail is discussion in  Supplementary Material Sec.~\ref{sec:appendix robust coherent state prep}.

\paragraph{Readout error}
Another predominant error in practice comes from state measurement. For example, atom loss and imperfect re-trapping can introduce measurement errors in Rydberg atom quantum simulators. More specifically, a ground state (\( \ket{0} \)) can be misidentified as a Rydberg state (\( \ket{1} \)) due to atom loss, with probability \( P_{\text{loss}} = \mathbb{P}(1|0) = 0.01 \); and a Rydberg state (\( \ket{1} \)) can be misread as a ground state (\( \ket{0} \)) due to imperfect anti-trapping, with probability \( P_{\text{anti-trapping}} = \mathbb{P}(0|1) = 0.08 \) \cite{wurtz2023aquila}.

We mitigate the readout error via the confusion matrix, which stores the conditional probability of measurement outcome being the bit-string $b_i$ given the quantum state is $\ket{b_j}$, i.e.
\begin{equation}
    R : = [\mathbb{P}(b_i|b_j)_{b_i,b_j \in \{0,1\}^n}],
\end{equation}
for $n$ qubit system. Take the two-atom system as an example, the confusion matrix is \begin{equation*}
    R =
    \begin{bmatrix}
    \mathbb{P}(00|00) & \mathbb{P}(01|00) & \mathbb{P}(10|00) & \mathbb{P}(11|00) \\

    \mathbb{P}(00|01) & \mathbb{P}(01|01) & \mathbb{P}(10|01) & \mathbb{P}(11|01) \\

    \mathbb{P}(00|10) & \mathbb{P}(01|10) & \mathbb{P}(10|10) & \mathbb{P}(11|10) \\

    \mathbb{P}(00|11) & \mathbb{P}(01|11) & \mathbb{P}(10|11) & \mathbb{P}(11|11) \\
    \end{bmatrix}.
\end{equation*}
The readout error in the measured outcome vector $\mathbf{q} = (q_{00}, q_{01}, q_{10}, q_{11})^{\mathrm{T}}$ can be mitigated by applying a confusion matrix correction, given by
\begin{equation}
    \mathbf{p} = (R^{T})^{-1} \mathbf{q},
\end{equation}
where $\mathbf{p}$ denotes the corrected probability vector.

\paragraph{Time-dependent coherent error}
Time-dependent coherent error is the most challenging adversarial in the learning task. In previous methods that operate in the real space, such as robust phase estimation \cite{kimmel2015rpe} and periodic calibration \cite{arute2020periodiccalibration}, the relationship between the unknown parameters becomes entangled through complex trigonometric and inverse trigonometric functions. During inference, this leads to a highly non-convex optimization problem involving both parameters. Consequently, an error in estimating one parameter can severely affect the accuracy of the other. However, our learning algorithm does the inference on the Fourier space, decoupling the multiple parameters into largely orthogonal components. To be more specific, the estimators designed in our learning algorithm are based on the relation between the Fourier coefficients and the parameters in each invariant subspace as follows,
\begin{equation}\label{eq:decouple}
    c_k \approx i\theta e^{-i(2k+1)\zeta}, \forall k = 0,\cdots,d-1.
\end{equation}
Here we completely separate two parameters, where $\theta$ depends on the amplitude of the Fourier coefficients and the phase of the coefficients implies $\zeta$. Further, the desired estimators of each invariant subspace can be inferred with the mapping Eq.~\eqref{eqn: invariant subspace relation}, which are also decoupled under small-angle approximation \footnote{The Eq.~\eqref{eqn: invariant subspace relation} maps the estimated pair $(\theta,\zeta)$ to the time-integral of our desired estimators, we can choose evolution time properly to approximately estimate $a^T$ from $\theta$ and $m^k$ via $\zeta$.}. As a result, this decoupling allows accurate estimation of one of them (e.g. $c_{ij}$ term) while the other is suffering from time-dependent coherent error (e.g. the drifting error on the local field $a_i$). To model the coherent error, we consider that there exists a $\gamma$ shift on the local fields, resulting in the coherent change in Hamiltonian to \(
H_i = (a_i+\gamma a_i) X_i + \sum Z_i Z_j 
\) in the simulation. 
\begin{figure}[h!]
    \centering    \includegraphics[width=\linewidth]{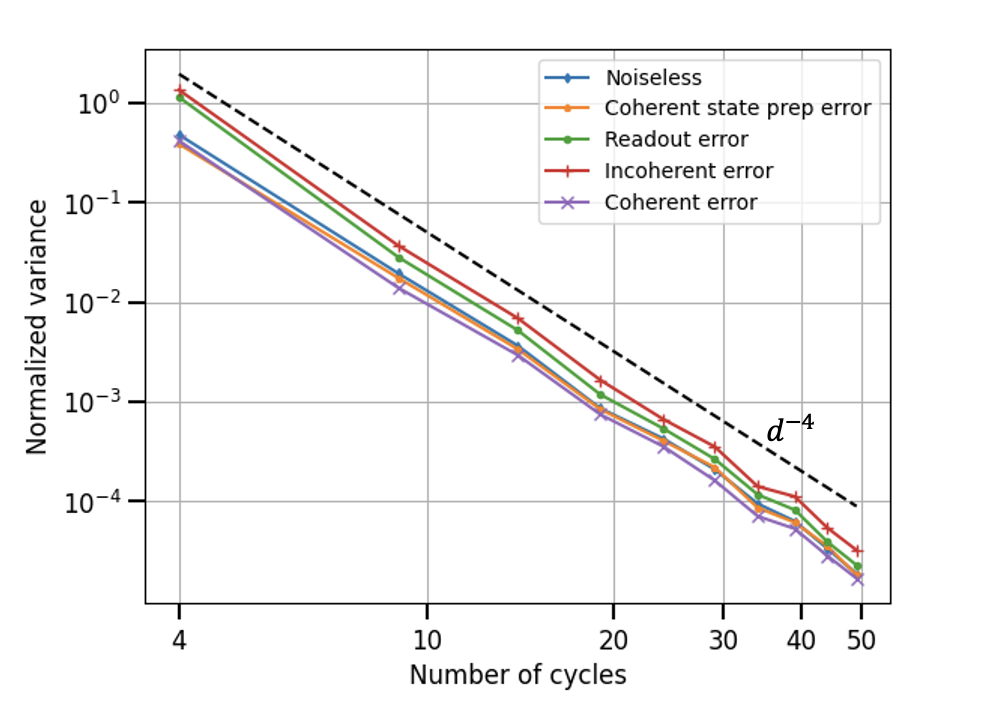}  \caption{Robustness test against coherent state preparation error(orange), readout error(green), depolarizing error(red) and time-dependent coherent error(purple) with the number of samples $N_{\text{shot}} = 10^5$ and the number of boostraps $N_{\text{bootstrap}}  = 10^3$. The noise model is chosen based on realistic noise model for Rydberg system \cite{wurtz2023aquila}, where the readout error is $(p_{\text{loss}},p_{\text{anti-trapping}}) =(0.01,0.08)$, the depolarizing noise has fidelity $\alpha = 0.8$, coherent state preparation error $\alpha = 0.01$ and the coherent error on $X_i$ is $\gamma = 10\%$. The normalized variance is defined as the relative variance $\text{Var}(\hc)/c^2$. The simulation results exhibit $\propto d^{-4}$ scaling in both noisy and noiseless scenarios, suggesting the robustness of the learning algorithm against experimental noise.}
    \label{fig:robustness test}
\end{figure}

To conclude the discussion on the robustness of our learning algorithm, we conducted a comprehensive study on the robustness of our learning algorithm under various realistic noise sources. The simulation results are presented in Figure.~\ref{fig:robustness test}, which showcases the algorithm’s performance within an invariant subspace of a 10-qubit all-to-all connected system governed by the Hamiltonian in Eq.~\eqref{eqn:full H}. Due to the parallelized structure of the algorithm, its performance remains consistent across systems of different sizes. All curves in the figure exhibit a $1/d^4$ scaling, consistent with Theorem~\ref{thm:heisenberg limit}, indicating that the algorithm maintains its effectiveness in both noiseless and noisy environments. These results confirm the algorithm’s robustness against readout errors, state preparation errors, depolarizing noise, and time-dependent coherent errors.

\subsubsection{Experimental proposal}\label{sec:experiment proposal}
The algorithm’s robustness against all major sources of realistic error highlights its potential for practical deployment on current NISQ devices to learn device-specific parameters. In this section, we provide a detailed implementation strategy for the Rydberg atom platform, including the experimental procedure, parameter settings, and simulation results. While the proposal is specifically tailored for learning two-atom Rydberg interactions, it can be naturally extended to systems with more atoms following Algorithm \ref{alg:qspe-analog}. Additionally, the same framework can be applied to superconducting platforms for learning XY couplings with the parallel-learnable Hamiltonain discussed in Supplementary Material Sec.~\ref{SM:pl hamiltonian def}.
    \begin{figure*}[t]\centering\includegraphics[width=0.7\linewidth]{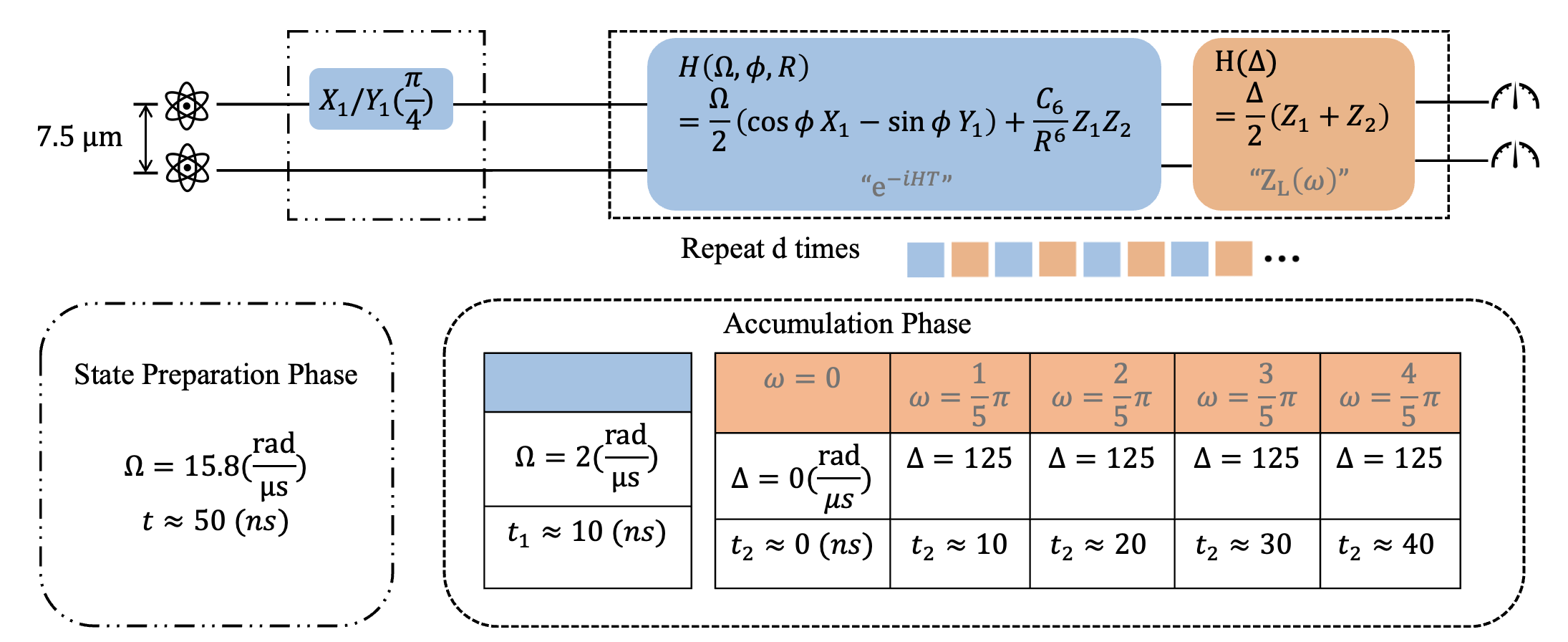}
        \caption{Physical implementation on the Rydberg quantum simulator. The first dotted block is for state preparation, where $\ket{+}_l$ and $\ket{i}_l$ should be prepared according to Supplementary Material Sec.~\ref{sec:implement detials for two qubit}. The second dotted block evolves $d$-repetition of target quantum dynamics and $Z_l$ phase accumulation. The whole evolution should run for $\omega = \frac{j}{2d-1}\pi, j = 0,\cdots,2d-2$. $d=3$ is picked here for reference.}
        \label{fig:exp}
    \end{figure*}
\paragraph{Experiment procedure}
Consider the scenario for the Rydberg Hamiltonian involving two atoms with the corresponding Hamiltonian
\begin{equation}
    H(\Omega,\phi,R) = \frac{\Omega}{2}(\cos\phi X_1 - \sin\phi Y_1) + c_{12}Z_1Z_2,
\end{equation}
where $\Omega$ is the Rabi frequency, $\phi$ is the phase and $c_{12} =\frac{C_6}{R^6}$ is the interaction strength. This is a general version of the full Hamiltonian given in Eq.~\eqref{eqn:full H} with phase $\phi$ between $X_1$ and $Y_1$, we can set the $\phi = 0$ to get back to Eq.~\eqref{eqn:full H} as the phase information is irrelevant here. For this Hamiltonian, we are interested in inferring the atom-atom distance $R$.

Similarly, consider the time-evolution operator for time $T$, where the accumulated angles are $\int_0^T\frac{\Omega}{2}dt = a$ and $\int_0^Tc_{12}dt = b$. The corresponding mapping for accumulated parameters is,
\begin{equation}\label{eqn:mapping realization}
    \begin{cases}
        \tan(\zeta) =  \frac{b}{\sqrt{a^2+b^2}}\tan(\sqrt{a^2+b^2})\\
        \tan(\chi) = \tan(\phi)\\
        \sin^2(\theta) = (\frac{a}{\sqrt{a^2+b^2}})^2\sin^2(\sqrt{a^2+b^2})
    \end{cases}
\end{equation}
Dividing evolution time T gives back to our desired terms, i.e., $(\Omega,c_{12}) = (a/T,b/T)$.
\begin{table*}[t]
    \centering
    \begin{tabular}{|c|c|c|c|}
    \hline
        atom distance $R$ & coefficient $V = \frac{C_6}{R^6}$& evolution time $t$ & effective integral $b = \int_0^T\frac{C_6}{R^6}dt $ \\
    \hline
        $7.16\mu m$ & $40~{rad}/{\mu s}$ & $1e-03 \mu s$ & $0.04~rad$\\
    \hline
    $7.52\mu m$ & $30~{rad}/{\mu s}$&$1e-03 \mu s$ &$0.03~rad$\\
    \hline
    $8.04\mu m$ & $20~{rad}/{\mu s}$&$1e-03 \mu s$ & $0.02~rad$\\
    \hline
    \end{tabular}
    \caption{Table of benchmark cases for testing the learning algorithm on Rydberg device. The evolution time can be chosen flexibly while satisfying the hardware constraint. The parameters listed here are used in simulation for reference.}
    \label{tab:benchmark}
\end{table*}

The physical implementation of the learning algorithm for the Rydberg system can be realized by considering the two main components of Algorithm~\ref{alg:qspe}.

1. Quantum part: The experimental realization of the algorithm is shown in the Figure.~\ref{fig:exp}. By adding up the time take for two phases in the Figure.~\ref{fig:exp}, the total evolution time required on the device for one tunable $\omega$ point can be roughly analyzed as,
    \begin{equation}
        T_{\text{total}} = t_{\text{state prep}}+t_{\text{repetition}} = t+(t_1+t_2)\times d \approx 180 ns.
    \end{equation}
    As a result, the total evolution time for all $(2d-1)$ sampled $\omega$ for the complete inference steps is then
    \begin{equation}
        T_{\text{round}} = 2\times(2d-1)\times T_{\text{total}} \approx 1.8\mu s,
    \end{equation}
    where $2$ stands for two different initial states and $(2d-1)$ is for different tunable $\omega$ terms. \\
    
2. Classical post-processing: Following the Algorithm \ref{alg:qspe}, the estimators $(\hat{\theta},\hat{\zeta})$ converts back to the accumulated angles $(\hat{a},\hat{b})$. The estimator of the atom distance $R$ follows
\begin{equation}\label{eqn:convertion}
    \hat{R} = (\frac{C_6}{\hat{b}/T})^{1/6},
\end{equation}
where $T$ is the evolution time and $C_6 = 5,420,503 ~{\mu m^6rad}/{\mu s}$ is a constant. Notably, under the proposed parameters, the total evolution time per experimental round is approximately $T_{\text{round}}\approx 1.8 \mu s$, which is significantly shorter than the system's $T_1$ and $T_2$ time \cite{wurtz2023aquila}, highlighting the algorithm’s feasibility for near-term implementation.

\paragraph{Parameter setting}
The Table \ref{tab:benchmark} summarizes the parameters used for testing on Rydberg devices. The selected inter-atomic distances are motivated by the findings in \cite{dag2024distanceerror}, which show significant deviations between experimental performance and theoretical predictions at those distances. Set the local field $\Omega$ such that the accumulated angle $a =\int_0^T\frac{\Omega(t)}{2}dt = 0.01~rad$, we ensure that $d\theta \ll 1$, under which the variance of $\hat{b}$ scales with the repetition depth as $\mathrm{Var}(\hat{b}) \sim O(1/d^4)$. Further, since the atom-atom interaction strength in the Rydberg system is in the form of
\(
    b(R) = C_6/R^6, 
\)
the variance relation between them is
\begin{equation}\label{eq:varaince relation}
    \var(b(R)) \approx b(R)'\var(R) = (\frac{6b(R)}{R})^2\var(R).
\end{equation}
As a result, the scaling between the variance of $\hat{R}$ and repetition depth $d$ is also $\var(\hat{R})\sim O(1/d^4)$.

\paragraph{Simulation results}
For the parameters listed in the Table.~\ref{tab:benchmark}, we simulate the performance of our learning algorithm according to the experiment procedure discussed in the previous section. The results are shown in the Figure.~\ref{fig:distance}, where the atom distances are $(R_1,R_2,R_3) = (7.16\mu m, 7.52 \mu m, 8.04 \mu m)$. The figure shows the identical scaling with respect to $\hat{b}$ and $\hat{R}$ satisfying the relation in Eq.~\eqref{eq:varaince relation}.
    \begin{figure*}[t]\centering\includegraphics[width=0.9\linewidth]{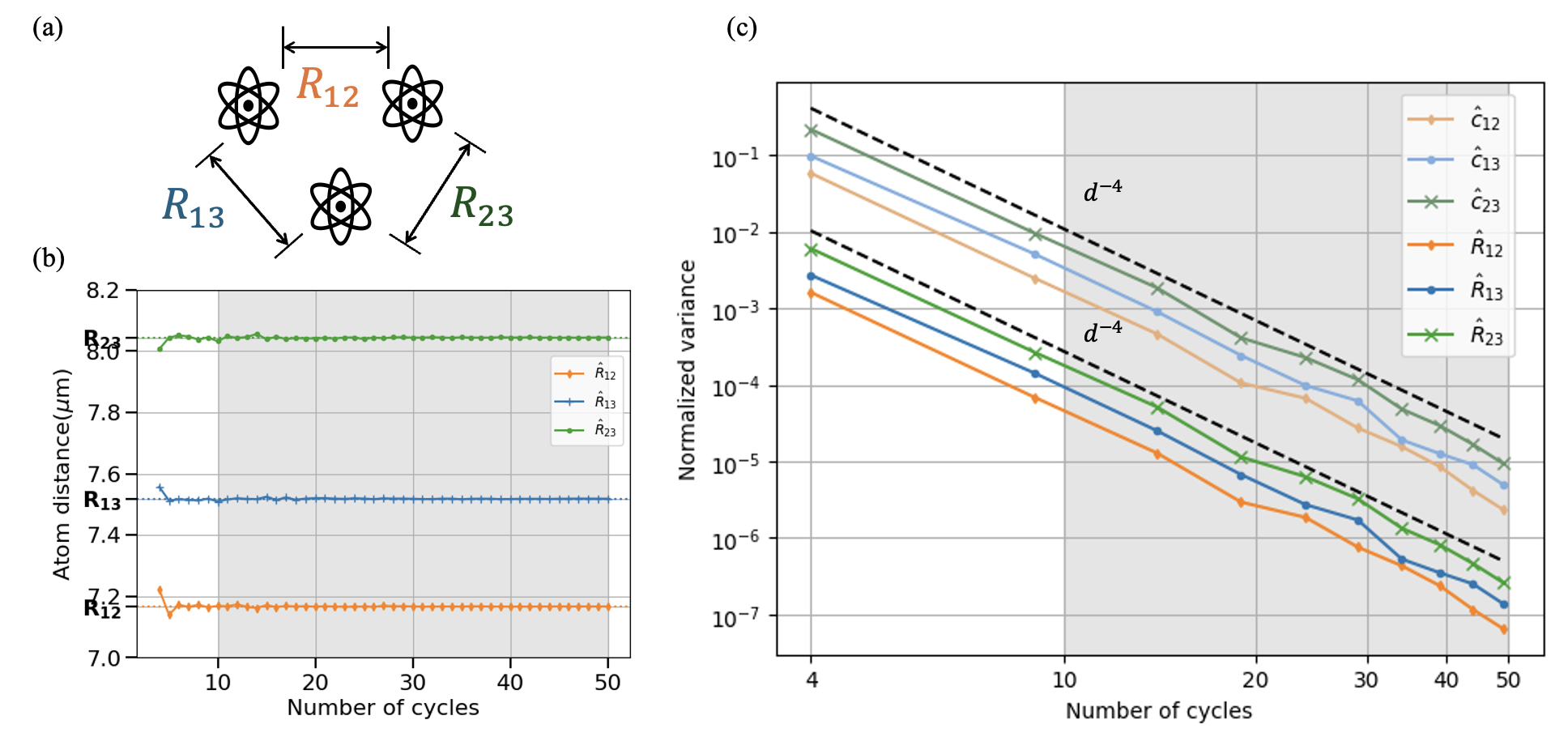}
        \caption{Application of the learning algorithm on Rydberg quantum simulators for atom distance. (a) Atom distance learning with three different distance ($R_{12} = 7.16\mu m, R_{13} = 7.52 \mu m, R_{23} = 8.04\mu m$) (b) The estimated distance of three estimators  ($\hat{R}_{12},\hat{R}_{13},\hat{R}_{23}$). The bold x-ticks are the true value ($R_{12},R_{13},R_{23}$). After $d = 10$ cycles, the estimations of all distances are consistently correct. (c) The upper block of curves (light-green, light-blue, light-orange) represents the variance of the estimators $(\hat{c}_{12},\hat{c}_{13},\hat{c}_{23})$ with respect to the number of cycles $d$, where $c_{ij}$ is the parameter of term $Z_iZ_j$ in Eq.~\eqref{eqn:full H}. While the lower block of curves(green, blue, orange) is the variance of atom distance estimators ($\hat{R}_{12},\hat{R}_{13},\hat{R}_{23}$) with respect to $d$ after conversion given in Eq.~\eqref{eqn:convertion}. The simulation results exhibit $\propto d^{-4}$ scaling, in agreement with Theorem \ref{thm:heisenberg limit} for all estimators.}
        \label{fig:distance}
    \end{figure*}
Besides, we notice from the simulation result that with the variance scaling as $O(1/d^4)$, an accurate estimation of the atom distance only requires shallow depths of dynamics repetition. From the Figure.~\ref{fig:distance}, to learn the distance with standard deviation $\delta R \approx 10^{-7}$, i.e. normalized variance $\frac{\delta R}{R} \approx 0.01$, only $d = 10$ repetition is needed.

\section{\label{sec:discussion} Conclusion}
In this work, we present the first in-situ Hamiltonian learning algorithm that is robust against a broad range of systematic and environmental noise. By exploiting the system Hamiltonian’s parallel structure, the algorithm achieves a quadratic speedup in sample complexity relative to the number of parameters compared with existing approaches for general Hamiltonian models \cite{huang2023learning,hu2025ansatz,ma2024learning,li2024heisenberg,ni2024quantum}. This demonstrates the potential of leveraging known Hamiltonian structures to efficiently certify quantum devices. Another important advantage of our learning algorithm is its ability to perform in-situ learning: all couplings remain active during the process, allowing us to capture crosstalk effects that arise only in the full many-body setting. This leads to a more accurate characterization of noise in current quantum systems compared to previous pairwise calibration approaches~\cite{andersen2025thermalization,dag2024distanceerror,bernien2017probing}. Looking ahead, this capability could be extended to enable calibration protocols dealing with correlated noise, provide deeper insights into non-Markovian noise processes, and guide the design of error-mitigation strategies that consider many-body interaction effects.

Besides, the algorithm attains the optimal precision and remains resilient to SPAM, decoherence, and time-dependent coherent errors. These results highlight its potential for practical deployment on NISQ devices, enabling high-precision learning of structured Hamiltonians in realistic settings. We provide the end-to-end proposal for learning distance error for Rydberg atom arrays. Unlike the previous approaches, which are only capable to learn distance pairwisely, we introduce the first in-situ learning algorithm for Rydberg atom distance learning. This feature is particularly crucial for Rydberg systems, where couplings cannot be naturally switched off without physically moving the atoms. Moreover, since our approach does not rely on long-time evolutions and is applicable to both hybrid and fully analog devices, the proposal is directly compatible with current experimental platforms. Our method is also applicable to other platforms with a parallel-learnable Hamiltonian with no additional overhead. It will also be interesting to check whether other experiment-induced Hamiltonians can be converted to parallel-learnable via realizable unitary transformations \cite{cubitt2018universal}. Additionally, the broader question of how prior structural knowledge can further reduce learning costs remains an open and valuable direction for future research. A comprehensive evaluation across various physical devices, through either simulation or real quantum experiments, is also left for future investigation.

\begin{acknowledgments}

We thank Grecia Castelazo and Srijan Nikhar for helpful discussions.  S.L. and X.W.~were supported by Air Force Office of Scientific Research under award number FA9550-21-1-0209, the U.S. National Science Foundation grant CCF-1942837 (CAREER), and a Sloan research fellowship. M. N. is
supported by the U.S. National Science Foundation grant CCF-2441912(CAREER), Air Force Office of Scientific Research under award number FA9550-25-1-0146, and   the U.S. Department of Energy,  Office of  Advanced Scientific Computing Research  under Award Number DE-SC0025430.

\end{acknowledgments}

\bibliography{main}


\clearpage

\onecolumngrid
\begin{center}
\textbf{\large Supplementary Material}
\end{center}
   
\setcounter{section}{0}
\setcounter{equation}{0}
\setcounter{figure}{0}
\setcounter{table}{0}
\setcounter{page}{1}
\renewcommand{\theequation}{S\arabic{equation}}
\renewcommand{\thefigure}{S\arabic{figure}}
\renewcommand{\thetable}{S\arabic{table}}

\vspace{.5cm}
\section{Parallel-learnable Hamiltonian}\label{SM:pl hamiltonian def}
In this section, we provide additional details on the structure of parallel-learnable Hamiltonians and the construction of corresponding unitary transformations. As described in the main text, an $n$-qubit Hamiltonian is said to be \emph{parallel-learnable} if it admits a decomposition into \( 2^{n-1} \) invariant subspaces, under a specific basis. This structure enables the parallel estimation of multiple parameters by tackling them within independent \( 2 \times 2 \) subspaces simultaneously. To obtain such a block-diagonal form, we apply a unitary transformation \( U \) that maps the Hamiltonian from the computational basis to a basis where the invariant subspaces are presented. In principle, any Hamiltonian \( H \) can be brought into a block-diagonal form consisting of non-diagonal \( 2 \times 2 \) blocks through a specially designed transformation \( U := R V \), where \( V \) is the eigenbasis matrix that diagonalizes \( H \), and \( R \) is a block-wise rotation that converts the diagonal \( 2 \times 2 \) blocks into non-diagonal ones. This transformation results in a Hamiltonian of the form described in Eq.~\eqref{eqn:pl H} of the main text. In the following, we elaborate on this construction for general Hamiltonians.

Consider the diagnolization of the Hamiltonain $H$, i.e.,
\[H = V\Lambda V^{\dagger},\] where \( V \) is the eigenbasis matrix, consisting all the eigenvectors of $H$ as the columns. Changing the basis of $H$ from computational basis to eigenbases results in the diagonal structure,
\begin{equation}
    \Lambda = V^{\dagger}HV = \oplus_{k = 1}^{2^{n-1}}\Lambda_k, ~~\text{where}~~ \Lambda_k = \begin{bmatrix}
        \lambda_k^{1} & 0 \\
        0 & \lambda_k^{2}
    \end{bmatrix} ~~\text{and} ~~ \{\lambda_k^{i}\}~~\text{are eigenvalues of }H.
\end{equation}
The resultant matrix $\Lambda$ after basis change is a direct sum of $2 \times 2$ diagonal matrices, which does not fit in the definition. So we define another unitary transformation matrix,
\begin{equation}
    R = \oplus_{k = 1}^{2^{n-1}}R_k(\theta_k), ~~\text{where}~~ R_k(\theta_k) = \begin{bmatrix}
        \cos\theta_k & -\sin\theta_k\\
        \sin\theta_k & \cos\theta_k
    \end{bmatrix}. 
\end{equation}
$R$ is a block-diagonal rotation matrix acting on each $2\times2$ block. So after applying the unitary transformation $R$, the matrix is transformed to
\begin{equation}
    H' = R\Lambda R^{\dagger} = \oplus_{k = 1}^{2^{n-1}} H_k,~~\text{where}~~H_k = \begin{bmatrix}
        \lambda_k^1\cos^2\theta_k+\lambda_k^2\sin^2\theta_k & (\lambda_k^1-\lambda_k^2)\cos\theta_k\sin\theta_k\\
        (\lambda_k^1-\lambda_k^2)\cos\theta_k\sin\theta_k&  \lambda_k^1\sin^2\theta_k+\lambda_k^2\cos^2\theta_k
    \end{bmatrix}.
\end{equation}
$H_k$'s are non-diagonal if $\lambda_k^1 \ne \lambda_k^2$ and $\theta_k \notin\{0,\pi/2\}$. As a result, the overall unitary transformation \( U := R V \) can transform any Hamiltonian \( H \) into the desired block structured form. It is important to note that this transformation is not unique and there exist multiple choices of unitary transformations that achieve the same structural transformation. The particular construction presented here relies on the eigenbasis matrix \( V \), which may be challenging to implement on current quantum hardware. To ensure the practicality of our learning algorithm, we therefore restrict the definition of parallel-learnable Hamiltonians to those admitting realizable unitary transformations.

Here we present some examples of parallel-learnable Hamiltonians with their corresponding realizable unitary transformations. These examples demonstrate how certain structured Hamiltonians can be transformed into a direct sum of \( 2 \times 2 \) non-diagonal blocks using unitary operations that are feasible to implement on current quantum hardware. For each case, we explicitly construct the unitary transformations.

\paragraph{Hadarmad-transformed} Consider the Hamiltonian in the following form
\begin{equation}
    H = \frac{1}{4}\begin{bmatrix}
         2a + 2b + 2c + 2d   &     2a - 2b         &     2c - 2d & 0 \\
          2a - 2b    &    2a + 2b - 2c - 2d      &      0 &                -2c+2d\\
          2c - 2d          &       0     &       2a + 2b + 2c + 2d   &     2a-2b\\
           0          &       -2c + 2d           &   2a - 2b       & 2a + 2b-2c-2d
    \end{bmatrix},
\end{equation}
it does not have clear invariant subspaces under the computational basis. However, it is still parallel-learnable since under the transformation
\begin{equation}\label{eq:hadamard transformation}
    U = \frac{1}{4}\begin{bmatrix}
        1 & 1 & 1 &1\\
        1 & -1 & 1 & -1\\
        1 & 1 & -1 & -1\\
        1 & -1& -1 & 1
    \end{bmatrix},
\end{equation}
it can be written in the direct sum of $2  \times 2$ non-diagonal matrix, i.e.,
\begin{equation}
    U^{\dagger}HU = \begin{bmatrix}
        a & c & 0 & 0\\
        c & b & 0 & 0\\
        0 & 0 & a & d\\
        0 & 0& d & b
    \end{bmatrix}.
\end{equation}
Also, the unitary transformation $U$ is easy to implement using just Hadarmard gates, i.e., $U = \text{H}\otimes\text{H}$.

\paragraph{Permutation-transformed} 
Consider the Hamiltonian of the general XY model with local $Z$ term for two qubits as an example, i.e.,
\begin{equation}
    H = g_xX_1X_2+g_yY_1Y_2+g_1Z_1+g_2Z_2 = \begin{bmatrix}
        g_1+g_2 & 0 & 0 & g_x+g_y\\
        0 & g_1-g_2 & g_x+g_y & 0\\
        0 & g_1-g_2 & -g_1+g_2 & 0\\
        g_x-g_y & 0 & 0 & -g_1-g_2
    \end{bmatrix},
\end{equation}
it cannot be written as a direct sum in the normal computational bases. However, it has clear invariant subspaces under the computational basis. Therefore, the unitary transformations to make it into a direct sum are just permutations. Under the permutation transformation,
\begin{equation}
    U = \begin{bmatrix}\label{eqn:permutation transformation}
        1 & 0 & 0 & 0\\
        0 & 0 & 1 & 0\\
        0 & 0 & 0 & 1\\
        0 & 1& 0 & 0
    \end{bmatrix},
\end{equation}
the resulting matrix will be
\begin{equation}
     U^{\dagger}HU = \begin{bmatrix}
         g_1+g_2 &  g_x-g_y & 0 & 0\\
         g_x-g_y & -g_1-g_2 & 0 & 0\\
        0 & 0 & g_1-g_2 &  g_x+g_y\\
        0 & 0&  g_x+g_y & -g_1+g_2
    \end{bmatrix}.
\end{equation}
This $U$ transformation can be realized with CNOT gates, i.e., $U = \text{CNOT}_{1 \rightarrow 0} \cdot \text{CNOT}_{0 \rightarrow 1}$. This Hamiltonian is commonly used to approximate the superconducting qubit Hamiltonian and interactions in Trapped Ion gates~\cite{kjaergaard2020superconducting,siddiqi2021engineering,msgate}

\section{Quantum Signal Processing Estimation}\label{sec:QSPE}
Parallel structure of the time-evolution operator allows the usage of the QSPE method \cite{dong2025qspe} in different $2\times2$ invariant subspace simultaneously. In this section, we briefly introduce the idea of the QSPE method. The QSPE method is designed for estimating the angles of the gate containing a two-level invariant subspace. The unitary gate restricted to the invariant subspace is parametrized as
\begin{equation}
    U = \begin{bmatrix}
        \cos(\theta) e^{-i\zeta} & -i\sin(\theta)e^{i\chi} \\
        -i\sin(\theta)e^{-i\chi} & \cos(\theta)e^{i\zeta} 
    \end{bmatrix}\notag,
\end{equation}
where the $\theta$ is the swap angle, $\zeta$ is the phase difference and $\chi$ is the phase accumulated during the swap. Denote the $\{\ket{0_l},\ket{1_l}\}$ as the logical basis of the invariant subspace, the quantum circuits utilize in QSPE methods are constructed as follows. First, the Bell state $\ket{+_l}$ and $\ket{i_l}$ is prepared. Then $d$ repetitions of the target gate $U$ followed by a tunable $Z$ rotation with respect to parameter $\omega$ are applied to the circuit. Finally, the circuit is measured in computational basis for the transition probability $p_x$ (for initial state $\ket{+_l}$) and probability $p_y$ (for initial state $\ket{i_l}$). These steps are used to amplify the angle of $\theta$ and $\zeta$ to reach the Heisenberg limit in the following inference step.

After the quantum stage, the QSPE method focuses on a reconstructed function $h(\omega) = p_x -\frac{1}{2}+i(p_y-\frac{1}{2})$. It admits an approximated expansion, $ h(\omega) = \sum_{-d+1}^{d-1} c_ke^{2ik\omega}$ and if $d\theta << 1$, the Fourier coefficients 
\begin{equation} \label{eq:ck}
   c_k \approx i\theta e^{-i\chi}e^{-i(2k+1)\zeta}\quad\text{with}~k =0,\cdots,d-1.
\end{equation}
To characterize the information completely, sampling the reconstructed function $h$ on $2d-1$ distinct $\omega$ points is required from the Fourier expansion. In QSPE method, the set of $\omega$ points is chosen with equal space for the quantum stage. Moreover, two amplified angles $\theta$ and $\zeta$ are completely decoupled in Eq.~\eqref{eq:ck}. As a result, the complicated inference problem becomes two independent linear regression problems with respect to the amplitude and phase of Fourier coefficients, leading to the estimators in QSPE method as follows.
\begin{equation}
    \hat{\theta} = \frac{1}{d}\sum_{k = 0}^{d-1}|c_k| \quad, \hat{\zeta} = \frac{1}{2}\frac{\vec{\mathbb{I}}^T\mathcal{D}^{-1}\vec{\Delta}}{\vec{\mathbb{I}}^T\mathcal{D}^{-1}\vec{\mathbb{I}}};
\end{equation}
where $\mathbb{I}$ is an all-one vector and $\mathcal{D}$ is the $(d-1)\times(d-1)$discrete Laplacian matrix as
\begin{equation}\label{eqn:def D}
    \mathcal{D} = \begin{bmatrix}
        2 & -1 & 0 &\cdots& 0\\
        -1 & 2 & -1 &\cdots& 0\\
        0 & -1 & 2 &\cdots &0\\
        \vdots & \vdots & \vdots & &\vdots\\
        0 & 0 & 0 & \cdots & 2
    \end{bmatrix},
\end{equation}
and $\mathbf{\Delta}: = (\Delta_0,\cdots,\Delta_{d-2})^T$ with the sequential phase difference is $\Delta_k : = \text{phase}(c_k\bar{c}_{k+1}).$ 

The variances of $(\hat{\theta},\hat{\zeta})$ in QSPE are
\begin{equation}
    \var(\hat{\theta}) \approx \frac{1}{8Nd^2}~,\var(\zeta)\approx \frac{3}{8Nd^4\theta^2},
\end{equation}
where $N$ is the number of shots and $d$ is the repetition depths. 

\section{Implementation Details of Algorithm \ref{alg:qspe}}\label{sec:implement detials for two qubit}
To apply the QSPE method~\cite{dong2025qspe} in Algorithm~\ref{alg:qspe}, adaptations of the original digital-gate-based QSPE protocols are required for physical implementation on quantum simulators. While the Rydberg atom quantum simulator is used as the implementation platform in this work, the proposed methods are equally applicable to other quantum simulation platforms. Since $\mathcal{U}_{\mathcal{B}}$ in Eq.~\eqref{eqn:UB_2 atom} is equivalent to the standard unitary $\mathcal{U}$ in Eq.~\eqref{eq:u gate} corresponding to the mapping from Eq.~\eqref{eqn: map}, thus adapting the QSPE to the defined logical subspace provides the characterization of unknown Hamiltonian coefficients. The adaptations contain: 1) prepare the logical $\ket{+_l}$ and $\ket{i_l}$ states; 2) applying logical control $Z_l$ term; 3) measuring transition probability with respect to the logical $\ket{0_l}$ state.

\textit{1) Adaptation 1}
To prepare logical Bell state $\ket{+_l} = \frac{1}{\sqrt{2}}(\ket{0_l}+\ket{1_l} )= \frac{1}{\sqrt{2}}(\ket{00}+\ket{10}$) and $\ket{i_l} = \frac{1}{\sqrt{2}}(\ket{0_l}+i\ket{1_l} )= \frac{1}{\sqrt{2}}(\ket{00}+i\ket{10}$), we make use of X and Y pulse correspondingly. First, we initialize both atoms in the ground state, i.e. $\ket{00}$, then consider X and Y rotations on the first atom, that is
\begin{align}
    e^{-i\theta X_1\otimes I}\ket{00} &= \cos\theta\ket{00}-i\sin\theta\ket{10}\label{eq:i_state_prep}\\ 
    e^{-i\theta Y_1\otimes I}\ket{00} &= \cos\theta\ket{00}+\sin\theta\ket{10};   \label{eq:plus_state_prep}
\end{align}
Thus, applying $-X(\frac{\pi}{4})$ leads to $\ket{i_l}$ and applying $Y(\frac{\pi}{4})$ leads to $\ket{0_l}$, where the negative sign is implemented by tuning phase $\phi = \pi$.

\textit{2) Adaptation 2}
Considering the invariant subspace $\mathcal{B}$, we want to effectively implement $Z_l$ acting on it. There are two options to achieve that, based on the control capability of the quantum system. In the analog-digital hybrid setting, applying global detuning to two atoms gives us the desired control unitary. Consider the global detuning $H = -\frac{\Delta(t)}{2}(Z_1+Z_2)$, the time-evolution operator is
\begin{align}
    \exp(-i\int_0^T H dt) &= \exp(-i\int_0^T -\frac{\Delta(t)}{2} dt (Z_1+Z_2)) \\
    & = \exp(-ic(Z_1+Z_2))\\
    & = \exp(-icZ_1)\exp(-icZ_2)\\
    & = \begin{bmatrix}
        e^{-i2c} & & &\\
        & 1 & & \\
         & & 1 & \\
         & & & e^{i2c}\\
    \end{bmatrix}\\
    & = \begin{bmatrix}
        e^{-ic} & & &\\
        & e^{ic} & & \\
         & & e^{ic} & \\
         & & & e^{i3c}\\
    \end{bmatrix} \text{up to a global phase}
\end{align}
where $c = \int_{0}^{T}\Delta(t)/2 dt$. Therefore, restricting to the invariant subspace gives us the control $Z_l$ with control parameter $c$.

In the fully analog setting, where we can not turn down the $ZZ$ interaction, then the control part can be provided by having global $Z$ interaction and atom-atom interaction together. Since $Z_1$,$Z_2$ and $Z_1Z_2$ commute, then we can write down the time evolution operator as:
\begin{align}
    \exp(-ic(Z_1+Z_2)-ibZ_1Z_2) &= \begin{bmatrix}
        e^{-ic} & & &\\
        & e^{ic} & & \\
         & & e^{ic} & \\
         & & & e^{i3c}
    \end{bmatrix} \begin{bmatrix}
        e^{-ib} & & &\\
        & e^{ib} & & \\
         & & e^{ib} & \\
         & & & e^{-ib}\\
    \end{bmatrix} \\
    &=  \begin{bmatrix}
        e^{-i(b+c)} & & &\\
        & e^{i(b+c)} & & \\
         & & e^{i(b+c)} & \\
         & & & e^{i(3c-b)}\\
    \end{bmatrix}
\end{align}
where $c = \int_{0}^{T}\Delta(t)/2 dt$ and $b = \int_{0}^{T}c_{12} dt$. Therefore, restricting to the invariant subspace gives us the control $Z_l$ with control parameter $b+c$.

\textit{3) Adaptation 3}
Measurements can also be done in the computational basis. We estimate the probability of getting $\ket{0_l} = \ket{00}$ by $\Pr(00) = N_{00}/N$, where $N_{00}$ is the number of string $00$ occurs and $N$ is the total shots.

Upon these three adaptations, the learning algorithm \ref{alg:qspe} can be implemented on the current Rydberg quantum device without any additional functionality required. 

\section{Robustness against coherent state preparation error}\label{sec:appendix robust coherent state prep}
Recall that the coherent error model considered here is the miscalibration in the state preparation step (discussed in the previous section in Adaptation 1. More specifically,
\begin{align}\label{eq:coherent state prep error}
      e^{-i(\theta+\alpha) X_1\otimes I}\ket{00} &= \cos(\theta+\alpha)\ket{00}-i\sin(\theta+\alpha)\ket{10} \notag\\ 
    e^{-i(\theta+\alpha)Y_1\otimes I}\ket{00} &= \cos(\theta+\alpha)\ket{00}+\sin(\theta+\alpha)\ket{10};\notag   
\end{align}
where $\alpha$ denotes the error rate. Inspired by the trick used for depolarizing noise, we construct a rescaled estimator reduce the error from first order to second order.
\begin{theorem*}[Restatement of Theorem III.1]
    Consider the coherent error Eq.~\eqref{eq:coherent state prep error} with error rate $\alpha$ in the state preparation step of Algorithm \ref{alg:qspe}, the difference between the noisy estimation $\hat{\theta'} = \frac{1}{d}\sum_{k = 0}^{d-1}|c_k'|$ and noiseless estimation $\hat{\theta} = \frac{1}{d}\sum_{k = 0}^{d-1}|c_k|$ can be bounded by
    \begin{equation}
      D =  |\hat{\theta}'-\hat{\theta}|\leq \sqrt{2}(d+1)^2\sin2\alpha\sin^2\theta+ (\cos2\alpha - 1)\theta,
    \end{equation}
    which scales with $O(\theta)$.
    The difference can be improved by introducing new type of estimator,i.e. scaled estimator $ \hat{\theta}_s' = \frac{1}{\cos2\alpha}\times \frac{1}{d}\sum_{k = 0}^{d-1}|c_k'|.$ The corresponding bound then reduces to 
    \begin{equation}
         D_s =  |\hat{\theta_s}'-\hat{\theta}|\leq  \sqrt{2}(d+1)^2\tan2\alpha \sin^2\theta,
    \end{equation}
    where scales only with $O(\theta^2)$.
\end{theorem*}
\begin{proof}
Consider the coherent error Eq.~\eqref{eq:coherent state prep error},the corresponding effective initial states $\ket{+_l}'$ and $\ket{i_l}'$ are in the forms of 
\begin{align}
   & \ket{+_l}' = \begin{bmatrix}
        \cos(\theta+\alpha)\\
        0\\
        \sin(\theta+\alpha)\\
        0
    \end{bmatrix} = \frac{1}{\sqrt{2}}\begin{bmatrix}
        \cos\alpha-\sin\alpha\\
        0\\
        \cos\alpha+\sin\alpha\\
        0
    \end{bmatrix} ,\\
& \ket{i_l}' = \begin{bmatrix}
        \cos(\theta+\alpha)\\
        0\\
        -i\sin(\theta+\alpha)\\
        0
    \end{bmatrix} = \frac{1}{\sqrt{2}}\begin{bmatrix}
        \cos\alpha+\sin\alpha\\
        0\\
        i(\cos\alpha-\sin\alpha)\\
        0
    \end{bmatrix} .\\
\end{align}
Recall the repeated circuits used in Algorithm \ref{alg:qspe} satisfy the form of
\begin{equation}
    U^{(d)} = 
    \begin{bmatrix}
    P^{(d)}(\cos\theta ) & i\sin\theta Q^{(d)}(\cos\theta)\\
    i\sin(\theta)Q^{(d)}(\cos\theta) & P^{(d)*}(\cos\theta)),
    \end{bmatrix}
\end{equation}
therefore the corresponding probability $p_X'$ and $p_Y'$ will be
\begin{align}
    p_X' & = |e^{i\frac{\chi+\pi+\varphi}{2}Z}\bra{0_l}U^{(d)}(\omega-\varphi,\theta)e^{-i(\omega+\frac{\chi+\pi+\varphi}{2})Z}\ket{+_l}'|^2\\
    & =|\frac{1}{\sqrt{2}}e^{i\frac{\chi+\pi+\varphi}{2}Z}\bra{0_l}U^{(d)}(\omega-\varphi,\theta)(e^{-i(\omega+\frac{\chi+\pi+\varphi}{2})}(\cos\alpha-\sin\alpha)\ket{0_l}+e^{i(\omega+\frac{\chi+\pi+\varphi}{2})}(\cos\alpha+\sin\alpha)\ket{1_l})|^2\\
    & = \frac{1}{2}|(e^{-i(\omega+\frac{\chi+\pi+\varphi}{2})}(\cos\alpha-\sin\alpha)P^{(d)}(\cos\theta )+i\sin\theta e^{i(\omega+\frac{\chi+\pi+\varphi}{2})}(\cos\alpha+\sin\alpha)Q^{(d)}(\cos\theta))|^2\\
    & = \frac{1}{2}\left((\cos\alpha-\sin\alpha)^2PP^{*}+\sin\theta^2(\cos\alpha+\sin\alpha)^2Q^2 \right) \nonumber\\
    & ~~ -\frac{1}{2}\left(i\sin\theta e^{-i2(\omega+\frac{\chi+\pi+\varphi}{2})}(\cos\alpha^2-\sin\alpha^2)QP+ i\sin\theta e^{i2(\omega+\frac{\chi+\pi+\varphi}{2})}(\cos\alpha^2-\sin\alpha^2)P^{*}Q \right)\\
    & = \frac{1}{2}\left((\cos\alpha-\sin\alpha)^2PP^{*}+\sin\theta^2(\cos\alpha+\sin\alpha)^2Q^2 + 2(\cos\alpha^2-\sin\alpha^2)\text{Re}(i\sin\theta e^{i(-2\omega-\chi+\varphi)}QP)\right) 
\end{align}
Then, we notice that 
\begin{align}
    & P^*P = \cos^2\left((d+1)\sigma\right)+\frac{\sin^2\left((d+1)\sigma\right)}{\sin^2\sigma}\sin^2\omega\cos^2\theta\\
    & \sin^2\theta Q^2 = \frac{\sin^2\left((d+1)\sigma\right)}{\sin^2\sigma}\sin^2\theta\\
    &  P^*P+\sin^2\theta Q^2 = 1\\
\end{align}
where $\cos\sigma = \cos\omega\cos\theta$.
Expanding the first term in the probability expression, 
\begin{align}
    (\cos\alpha-\sin\alpha)^2PP^{*}+\sin\theta^2(\cos\alpha+\sin\alpha)^2Q^2
    & = P^*P+\sin^2\theta Q^2 -2\cos\alpha\sin\alpha(P^*P-\sin^2\theta Q^2)\\
    & = 1-2\cos\alpha\sin\alpha(P^*P-\sin^2\theta Q^2).
\end{align}
Also, the minus term can be computed as 
\begin{align}
    P^*P-\sin^2\theta Q^2
    & = \cos^2\left((d+1)\sigma\right)+\frac{\sin^2\left((d+1)\sigma\right)}{\sin^2\sigma}(\sin^2\omega\cos^2\theta-\sin^2\theta)\\
    & = 1- \sin^2\left((d+1)\sigma\right) \frac{2\sin^2\theta}{\sin^2\sigma}\\
    & =  1- C_{\omega,d,\theta}\\
\end{align}
where we denote $C_{\omega,d,\theta} :=\frac{\sin^2\left((d+1)\omega\right)}{\sin^2\omega}$.

Thus, combing all above, the effective probability $p_X'$ is
\begin{align}
    p_X'&  = \frac{1}{2}[1-2\cos\alpha\sin\alpha(1-C_{\omega,d,\theta})]+ (\cos\alpha^2-\sin\alpha^2)\text{Re}(i\sin\theta e^{i(-2\omega-\chi+\varphi)}QP)\\
    & = \cos2\alpha p_X +[\frac{1-\cos2\alpha-\sin2\alpha(1-C_{\omega,d,\theta})}{2}],
\end{align}
where $p_X$ is the probability without noise.

Similarly, for the initial state $\ket{i_l}'$, we end up with 
\begin{align}
    p_Y'
    & = \frac{1}{2}\left((\cos\alpha+\sin\alpha)^2PP^{*}+\sin\theta^2(\cos\alpha-\sin\alpha)^2Q^2 + 2(\cos\alpha^2-\sin\alpha^2)\text{Im}(i\sin\theta e^{i(-2\omega-\chi+\varphi)}QP)\right) \\
    &  = \frac{1}{2}[1+2\cos\alpha\sin\alpha(1-C_{\omega,d,\theta})]+ (\cos\alpha^2-\sin\alpha^2)\text{Im}(i\sin\theta e^{i(-2\omega-\chi+\varphi)}QP)\\
    & = \cos2\alpha p_Y +([\frac{1-\cos2\alpha+\sin2\alpha(1-C_{\omega,d,\theta})}{2}].
\end{align}
Notice here that the effective probability is not simply constant shifted and scaled, so the trick for depolarizing channel does not work here.

Since 
$  h = p_X -\frac{1}{2}+i(p_Y-\frac{1}{2})$,
the effective $h'$ is
\begin{align}
    h' 
    & = p_X' -\frac{1}{2}+i(p_Y'-\frac{1}{2})\\
    & = \cos2\alpha h - \frac{\cos2\alpha+\sin2\alpha(1-C_{\omega,d,\theta})}{2}-\frac{\cos2\alpha-\sin2\alpha(1-C_{\omega,d,\theta})}{2}i.
\end{align}

Recall that the estimator in the QSPE algorithm \ref{alg:qspe} is of the form of:
\begin{equation}
     \hat{\theta} = \frac{1}{d}\sum_{k = 0}^{d-1}|c_k|,\label{eqn:theta_estimator}
\end{equation}
where the Fourier transform gives the desired coefficient for inferring our target parameters, i.e.
\begin{equation}
    c_k = \frac{1}{\pi}\int_{0}^{\pi} e^{-i(2k\omega)}hd\omega.
\end{equation}

\textbf{Original estimator}
Coherent error in state preparation will induce errors in probability $p_X$ and $p_Y$, resulting in errors in $h$. Thus, using the same estimator given in Eq.~\eqref{eqn:theta_estimator}, we can measure the performance of this estimator as follows.

First, we do the Fourier transform on $h'$ to get coefficients $c_k'$ for the original estimator.
\begin{align}
    c_k' & =  \frac{1}{\pi}\int_{0}^{\pi} e^{-i(2k\omega)}h' d\omega\\
    & = \frac{1}{\pi}\int_{0}^{\pi} e^{-i(2k\omega)}[\cos2\alpha h - \frac{\cos2\alpha+\sin2\alpha(1-C_{\omega,d,\theta})}{2}-\frac{\cos2\alpha-\sin2\alpha(1-C_{\omega,d,\theta})}{2}i]d\omega\\
    & = \cos2\alpha c_k - \frac{1}{\pi}\int_0^{\pi} e^{-i(2k\omega)}[\frac{\cos2\alpha+\sin2\alpha(1-C_{\omega,d,\theta})}{2}-\frac{\cos2\alpha-\sin2\alpha(1-C_{\omega,d,\theta})}{2}i] d\omega\\
    & = \cos2\alpha c_k - \frac{1}{\pi}\int_0^{\pi} e^{-i(2k\omega)}[ \frac{\cos2\alpha+\sin2\alpha}{2}+\frac{\cos2\alpha-\sin2\alpha}{2}i]-e^{-i(2k\omega)}[\frac{\sin2\alpha C_{\omega,d,\theta}}{2}+\frac{\sin2\alpha C_{\omega,d,\theta}}{2}i] d\omega\\
    & = \cos2\alpha c_k +\frac{\sin2\alpha}{2\pi}\int_{0}^{\pi} e^{-i(2k\omega)}(C_{\omega,d,\theta}+C_{\omega,d,\theta}i)d\omega
\end{align}
where $C_{\omega,d,\theta} = 2\sin^2\theta\frac{\sin^2[(d+1)\sigma]}{\sin^2\sigma} \leq 2\sin^2\theta(d+1)^2$ and $\cos\sigma = \cos\omega\cos\theta$.

Then the amplitude of the coefficients can be bounded as follows:
\begin{align}
    |c_k'| & \leq  \cos2\alpha|c_k| +\frac{\sin2\alpha}{2\pi}\int_{0}^{\pi} |e^{-i(2k\omega)}(C_{\omega,d,\theta}+C_{\omega,d,\theta}i)|d\omega\\
    &  = \cos2\alpha|c_k| +\frac{\sin2\alpha}{2\pi}\int_{0}^{\pi} \sqrt{2}|(C_{\omega,d,\theta})|d\omega\\
    & \leq \cos2\alpha|c_k| +\frac{\sin2\alpha}{\sqrt{2}}\max_{[0,\pi]} |(C_{\omega,d,\theta})|\\
    & \leq \cos2\alpha|c_k| +\sqrt{2}\sin2\alpha\sin^2\theta(d+1)^2\\
\end{align}

Therefore, the difference between the original estimator under coherent state preparation error 
$
    \hat{\theta}' = \frac{1}{d}\sum_{k = 0}^{d-1}|c_k'|
$
and the noiseless estimator 
$\hat{\theta} = \frac{1}{d}\sum_{k = 0}^{d-1}|c_k|$ will be:
\begin{align}
    |\hat{\theta}'-\hat{\theta}| &= \frac{1}{d}\sum_{k = 0}^{d-1} |c_k'|-|c_k|\\
    & \leq  \frac{1}{d}\sum_{k = 0}^{d-1} |c_k'-c_k|\\
    & =  \frac{1}{d}\sum_{k = 0}^{d-1} |c_k'-\cos2\alpha c_k+\cos2\alpha c_k-c_k|\\
    & \leq  \frac{1}{d}\sum_{k = 0}^{d-1} |c_k'-\cos2\alpha c_k|+|\cos2\alpha c_k-c_k|\\
    & \leq  \frac{1}{d}\sum_{k = 0}^{d-1} \sqrt{2}(d+1)^2\sin2\alpha\sin^2\theta+ \frac{1}{d}\sum_{k = 0}^{d-1}|\cos2\alpha c_k-c_k|\\
    & = \sqrt{2}(d+1)^2\sin2\alpha\sin^2\theta+ (\cos2\alpha - 1)\theta
\end{align}

\textbf{Scaled Estimator}
Inspired by the above analysis, we introduce a new estimator by scaling the original one with $1/\cos2\alpha$, i.e.
$ \hat{\theta}' = \frac{1}{\cos2\alpha}\times \frac{1}{d}\sum_{k = 0}^{d-1}|c_k'|.$

Then the corresponding bias will be:
\begin{align}
    |\hat{\theta}' - \hat{\theta}|& = \frac{1}{d}\sum_{k = 0}^{d-1}|\frac{1}{\cos2\alpha}c_k'| - |c_k|\\
    & \leq \frac{1}{d}\sum_{k = 0}^{d-1}|\frac{1}{\cos2\alpha}c_k' - c_k|\\
    &  = \frac{1}{d}\sum_{k = 0}^{d-1} \sqrt{2}\sin2\alpha\sin^2\theta(d+1)^2/\cos2\alpha\\
    & = \sqrt{2}(d+1)^2\tan2\alpha \sin^2\theta.
\end{align}
\end{proof}

\section{Variance computation for Algorithm \ref{alg:qspe-parallel}}\label{apx:var}
For each full Hamiltonian, we simultaneously work on $(n-1)$ invariant subspaces as defined in Eq.~\eqref{def:invariant subspace}. As a result, we should prepare the initial logical Bell states with respect to each invariant subspace and superpose them. Due to the normalization factor, the final state restricted to each invariant subspace gains a $1/\sqrt{(n-1)}$ factor compared to the standard Bell pair. Without loss of generality, let us focus on the full Hamiltonian $H_1$ and the $k$th invariant subspace of the Hamiltonian. Denote the transition probability with initial plus state as $p_k^M(\omega_j;\theta,\phi,\chi)$, the relation between the transition probability with superposition of the Bell pairs and the standard Bell pair is
\begin{equation}\label{eqn: transition relation}
    p_k^M(\omega_j;\theta,\phi,\chi) = p^B(\omega_j;\theta,\phi,\chi)/(n-1).
\end{equation}

Consider the sampling error in the system. Let $N$ be the number of measurements in one experiment with the set of parameters $(\omega_j;\theta,\phi,\chi)$, the i.i.d. measurement outcomes are sampled from a Multinomial distribution such that
\begin{equation}
    \sum_{i =1}^{(n-1)} p_i^M(b_i)+ p_i^M(1-b_i) = 1,
\end{equation}
where $b_i$ is the bitstring corresponding to the basis that spans the $i$th invariant subspace and $1$ is the all-one bitstring. 

Given $N$ large enough samples, the transition probability corresponding to each invariant subspace is estimated by 
\begin{equation}
    \hat{p}^{M}_k(\omega_j) = \frac{\# (x =b_k) }{N},
\end{equation}
where $b_k$ is the logical zero state for the $k$th invariant subspace. Focusing on this specific transition probability for the $k$th invariant subspace, the Multinomial distribution can be treated as a Bernoulli distribution with probability $p_k^M$. As a result, when the sample size $N$ is large enough, the central limit theorem leads to the normal distribution with mean $\mathbb{E}[\hat{p}^M_{k}(\omega_j)] = p_k^M$ and variance follows 
\begin{equation}
    \Sigma_k^2 = \frac{p_k^M(1-p_k^M)}{N} = \frac{1}{4N}-\frac{(p_k^M-\frac{1}{2})^2}{N}.
\end{equation}

Considering the standard Bell pair as the initial state for this specific invariant subspace, the transition probability is $p^B := (n-1)p_k^M$. Let the sample size still be $N$, the MLE for the transition probability will be $\hat{p}^B = \frac{\#(x =0)}{N}$ with expectation $\mathbb{E}(\hat{p}^B) = p^B$ and variance $\text{Var}(\hat{p}^B) = \frac{p^B(1-p^B)}{N}$. From \cite{dong2025qspe}, we notice the variance is bounded by $[\frac{1}{4N}-\frac{(d\theta)^2}{M}, \frac{1}{4N}]$. For the inference process, we utilize the sample data from the Multinomial senario to estimate the starndard transition probability following Eq.~\eqref{eqn: transition relation}, the estimator is defined as
\begin{equation}
    \hat{p}^M_{k,res} = (n-1)\hat{p}^{M}_k(\omega_j) = \frac{(n-1)\# (x =b_k) }{N},
\end{equation}
where $ \hat{p}^M_{k,res} $ denotes the rescaled probability. Then the expectation and variance value of $ \hat{p}^M_{k,res} $ can be computed as follows.
\begin{equation}
    \mathbb{E}( \hat{p}^M_{k,res} ) = (n-1)\mathbb{E}[\hat{p}^M_{k}] = (n-1)p_k^M = p^B;
\end{equation}
\begin{align}\label{eqn: rescale var}
    \text{Var}(\hat{p}^M_{k,res}) &= (n-1)^2 \text{Var}(\hat{p}^M_{k}) = (n-1)^2  \frac{p_k^M(1-p_k^M)}{N} = \frac{p^B(n-1-p^B)}{N},
\end{align}
where the last inequality follows Eq.~\eqref{eqn: transition relation}. Rewrite the Eq.~\eqref{eqn: rescale var}, we derive the upper and lower bounds of it as follows.
\begin{align}
    \text{Var}(\hat{p}^M_{k,res}) & = \frac{p^B(n-1-p^B)}{N}\\
    & = -\frac{(p^B-\frac{1}{2})^2}{N} -\frac{p^B}{N}+\frac{1}{4N}+\frac{(n-1)p^B}{N}\\
    & = -\frac{(p^B-\frac{1}{2})^2}{N} +
    (n-2)\frac{(p^B-\frac{1}{2})}{N}+\frac{(2n-3)}{4N}.\\
\end{align}
Given that $(p^B-\frac{1}{2})^2\leq (d\theta)^2$, we can bound the variance as
\begin{equation}\label{eqn:variance bound}
  \frac{(2n-3)}{4N}-[\frac{(d\theta)^2+(n-2)(d\theta)}{N}]\leq   \text{Var}(\hat{p}^M_{k,res}) \leq \frac{n-2}{N},
\end{equation}
where the left inequality matches when $(p^B-\frac{1}{2}) = -d\theta$ and right inequality matches when $(p^B-\frac{1}{2}) = \frac{1}{2}$. 

Then, we compute the Fourier coefficients of the transition probability $\{\hat{p}^M_{k,res}(\omega_j)\}$ from the experimental data. The FFT is performed on the reconstructed function for the $k$th invariant subspace, i.e. $\hat{h}_j:= \hat{p}_{X,res}^M+ i \hat{p}_{Y,res}^M -\frac{1+i}{2}$. The expectation value of this reconstructed function follows $\mathbb{E}[\hat{h}_j] = p^B_X+ip^B_Y-\frac{1+i}{2}$ admitting the analytical form. We compute the covariance matrix of the reconstructed function as follows. Let $u_j = \hat{h}_j-\mathbb{E}[\hat{h}_j]$, it holds that
\begin{align}
    & \text{for any}~j~, \mathbb{E}(u_j) = 0 , \mathbb{E}(|u_j|^2) = \text{Var}(\hat{p}_{X,res}^M)+\text{Var}(\hat{p}_{Y,res}^M) =\frac{p_X^B(n-1-p_X^B)}{N} +\frac{p_Y^B(n-1-p_Y^B)}{N}\\
    & \text{for any} ~ j\neq j', \mathbb{E}(u_ju_{j'}) = 0.
\end{align}
By doing FFT, the coefficients extracted are
\begin{equation}
\begin{bmatrix}
    \hat{c}_0\\
    \vdots\\
    \hat{c}_{d-1}\\
    \hat{c}_{-d+1}\\
    \vdots\\
    \hat{c}_{-1}
\end{bmatrix}
 = \frac{1}{2d-1}\Omega^{\dagger}
 \begin{bmatrix}
     \hat{h}_0\\
     \vdots\\
     \hat{h}_{2d-2}
 \end{bmatrix}, ~\text{where}~ \Omega_{jk} = e^{i\frac{2\pi jk}{2d-1}}.
\end{equation}

Assume the Fourier coefficients are approximately distributed following the complex normal distribution as
\begin{equation}
    \hat{c}_k  = c_k +\begin{cases}
        v_k, ~~~~~~~~~~~k= 0,\cdots, d-1;\\
        v_{2d-1+k},~~ k = -d+1, \cdots, -1,
    \end{cases}
\end{equation}
where $v_k$'s are complex normally distributed random variables. Next, we compute the mean and covariance matrix of the $v_k$’s, which offers insight into the error introduced by finite sampling directly. Since for all rescaled transition probabilities, they are be written as
\begin{equation}
    \hat{p}^M_{X/Y,res}(\omega_j) = p^B_{X/Y} + st(\hat{p}^M_{X/Y,res})\epsilon_i,
\end{equation}
where $\epsilon_i \sim N(0,1)$ and $\frac{(2n-3)}{4N}-[\frac{(d\theta)^2+(n-2)(d\theta)}{N}]\leq   \text{Var}(\hat{p}^M_{k,res}) \leq \frac{n-2}{N}$. Then the reconstructed function can also be interpreted as
\begin{equation}
    \hat{h}_j:= \hat{p}_{X,res}^M+ i \hat{p}_{Y,res}^M -\frac{1+i}{2} = p_X^B+ip_Y^B-\frac{1+i}{2}+ st(\hat{p}^M_{X/Y,res})\epsilon_i.
\end{equation}
By linearity, 
\begin{equation}
    v_k = \frac{1}{2d-1}(\Omega^{\dagger}u) = \frac{1}{2d-1}\sum_{j = 0}^{2d-2} \overline{\Omega_{kj}}u_j,
\end{equation}
where $u = [u_0 \cdots u_{2d-2}]^T$, the vector contains all $u_j = \hat{h}_j-\mathbb{E}[\hat{h}_j]$.

The expectation value is $\mathbb{E}(v_k) =\frac{1}{2d-1}\sum_{j = 0}^{2d-2} \overline{\Omega_{kj}}\mathbb{E}(u_j) = 0$, and the covariance is 
\begin{equation}
    \mathbb{E}(v_k\overline{v}_{k'}) = \frac{1}{(2d-1)^2} \sum_{j,j' = 0}^{2d-2}\overline{\Omega_{kj}}\Omega_{k'j'}\mathbb{E}(u_j\overline{u_{j'}}) = \frac{1}{(2d-1)^2} \sum_{j,j' = 0}^{2d-2}e^{i\frac{2\pi}{2d-1}(k'-k)j}\mathbb{E}(|u_j|^2) .
\end{equation}
When $k = k'$, the variance terms are
\begin{align}
    \mathbb{E}(|v_k|^2) & = \frac{1}{(2d-1)^2}\sum_{j = 0}^{2d-2} \mathbb{E}(|u_j|^2)\\
    &  = \frac{1}{(2d-1)^2}\sum_{j = 0}^{2d-2} -\frac{(p_X^B-\frac{1}{2})^2
    +(p_Y^B-\frac{1}{2})^2}{N} +
    (n-2)\frac{(p_X^B-\frac{1}{2})+(p_Y^B-\frac{1}{2})}{N}+\frac{(2n-3)}{2N}\\
    & = \frac{(2n-3)}{2N(2d-1)}+ \frac{1}{(2d-1)^2}\sum_{j = 0}^{2d-2} -\frac{(p_X^B-\frac{1}{2})^2
    +(p_Y^B-\frac{1}{2})^2}{N} +
    (n-2)\frac{(p_X^B-\frac{1}{2})+(p_Y^B-\frac{1}{2})}{N}.
\end{align} 
It can be bounded as
\begin{equation}\label{eqn: vk variance bd}
  \frac{(2n-3)-4(d\theta)^2-4(n-2)(d\theta)}{2N(2d-1)}  \leq \mathbb{E}(|v_k|^2)\leq  \frac{(2n-3)-4(d\theta)^2+4(n-2)(d\theta)}{2N(2d-1)}  ,
\end{equation}
where the conditions $(p_{X/Y}^B-\frac{1}{2})^2 \leq (d\theta)^2$ are used.

When $k \neq k'$, since $\sum_{j=0}^{2d-2}e^{i\frac{2\pi}{2d-1}(k'-k)j} =(2d-1)\delta_{kk'}$, the constant term $\frac{2n-3}{2N}$ vanishes. Thus,
\begin{align}
    |\mathbb{E}(v_k\overline{v}_{k'})|& = \left|\frac{1}{(2d-1)^2} \sum_{j = 0}^{2d-2}e^{i\frac{2\pi}{2d-1}(k'-k)j} \left(\frac{-[(p_X^B-\frac{1}{2})^2
    +(p_Y^B-\frac{1}{2})^2]}{N} +
    (n-2)\frac{(p_X^B-\frac{1}{2})+(p_Y^B-\frac{1}{2})}{N}\right)\right|\\
    & \leq \frac{1}{N(2d-1)^2} \sum_{j = 0}^{2d-2}\left|\left(\frac{-[(p_X^B-\frac{1}{2})^2
    +(p_Y^B-\frac{1}{2})^2]}{N} +
    (n-2)\frac{(p_X^B-\frac{1}{2})+(p_Y^B-\frac{1}{2})}{N}\right)\right|\\
    & \leq \frac{2(d\theta)^2+2(n-2)(d\theta)}{N(2d-1)}.
\end{align}

Consider the signal-to-noise ratio(SNR) for each Fourier coefficients
\begin{equation}
    \text{SNR}_k:=\frac{|c_k|^2}{\mathbb{E}(|v_k|^2)} \geq N(2d-1)\sin^2\theta(1-\frac{4}{3}(d\theta)^2(1+3d^3\theta^2))/(2n-4),
\end{equation}
this inequality comes from Theorem 12~\cite{dong2025qspe} and Eq.~\eqref{eqn: vk variance bd}.

When SNR is high, that is, $dM\theta^2/n >>1$, the decoupling of the $\theta$ and $\zeta$ dependence is well captured and for all $k = 0,\cdots, d-1,$ $d^5\theta^4 << 1$, the estimators are defined as
\begin{align}
    & \text{amplitude}(\hat{c}_k) \approx \theta +v_k^{(amp)};\\
    & \text{phase}(\hat{c}_k) = \frac{\pi}{2}-(2k+1)\zeta+v_k^{(pha)} ~~~\text{(up to 2}\pi-\text{periodicity)},
\end{align}
where $v_k^{(amp)}$ and $v_k^{(pha)}$ are normal distributed and can be approximately computed by $v_k^{(amp)} = \text{Re}(v_k)$, $v_k^{(pha)} = \text{Im}(v_k)/c_k(\theta)$ \cite{dong2025qspe}. Then we can write the covariance matrices as $\text{Cov}_{kk'}^{(amp)}:= \mathbb{E}(v_k^{(amp)}v_{k'}^{(amp)})$ and $\text{Cov}_{kk'}^{(pha)}:= \mathbb{E}(v_k^{(pha)}v_{k'}^{(pha)})$.

The covariance matrices can be computed from the real and imaginary components of the sampling error $v_k$ in Fourier coefficients. For any $k \neq k'$, 
\begin{equation}
    |\mathbb{E}(\text{Re}(v_k)\text{Re}(v_{k'})+\text{Im}(v_k)\text{Im}(v_{k'}))| = |\mathbb{E}(v_k\overline{v}_{k'})| \leq \frac{2(d\theta)^2+2(n-2)(d\theta)}{N(2d-1)}.
\end{equation}
For any $k$,
\begin{align}
 \frac{(2n-3)-4(d\theta)^2-4(n-2)(d\theta)}{2N(2d-1)}  &\leq  \mathbb{E}(\text{Re}^2(v_k)+ \text{Im}^2(v_k)) =  \mathbb{E}(|v_k|^2) \\
 &\leq  \frac{(2n-3)-4(d\theta)^2+4(n-2)(d\theta)}{2N(2d-1)} .
\end{align}

Then for any $k$,
\begin{align}
    |\mathbb{E}(\text{Re}^2(v_k)- \text{Im}^2(v_k))| & \leq  |\mathbb{E}(\text{Re}^2(v_k)- \text{Im}^2(v_k)+2i\text{Re}(v_k)\text{Im}(v_k))|\\
    & = |\mathbb{E}(v_k^2)| \leq \frac{1}{(2d-1)^2}\sum_{j = 0}^{2d-2} \mathbb{E}(|u_j^2|)  \\
    & = \frac{1}{(2d-1)^2}\sum_{j = 0}^{2d-2} |\text{Var}(\hat{p}_{X,res}^M)-\text{Var}(\hat{p}_{Y,res}^M) | \leq \frac{2(d\theta)^2+2(n-2)(d\theta)}{N(2d-1)}.
\end{align}
By the triangle inequality,
\begin{align}
    |\mathbb{E}(\text{Re}^2(v_k)) &- \frac{(2n-3)-4(d\theta)^2+4(n-2)(d\theta)}{4N(2d-1)} | \leq \frac{1}{2}|\mathbb{E}(\text{Re}^2(v_k)+ \text{Im}^2(v_k))\notag\\
    &- \frac{(2n-3)-4(d\theta)^2+4(n-2)(d\theta)}{4N(2d-1)}| + \frac{1}{2}   |\mathbb{E}(\text{Re}^2(v_k)- \text{Im}^2(v_k))| \\
    &\leq \frac{2(d\theta)^2+2(n-2)(d\theta)}{N(2d-1)}.
\end{align}
Similarly, for the imaginary part, we also have the above bound, i.e.,

\begin{equation}
     |\mathbb{E}(\text{Im}^2(v_k))- \frac{(2n-3)-4(d\theta)^2+4(n-2)(d\theta)}{4N(2d-1)}| \leq\frac{2(d\theta)^2+2(n-2)(d\theta)}{N(2d-1)}.
\end{equation}

Also, for any $k \neq k'$,
\begin{align}
    \mathbb{E}(\text{Im}(v_k)\text{Im}(v_{k'})) & \leq \frac{1}{2} |\mathbb{E}(\text{Re}(v_k)\text{Re}(v_{k'})+\text{Im}(v_k)\text{Im}(v_{k'}))| + \frac{1}{2} |\mathbb{E}(-\text{Re}(v_k)\text{Re}(v_{k'})+\text{Im}(v_k)\text{Im}(v_{k'}))| \\
    & \leq \frac{2(d\theta)^2+2(n-2)(d\theta)}{N(2d-1)}.
\end{align}
From \cite{dong2025qspe} Eq.~(82), when $d^3 \theta^2 << 1$ and $d\theta \leq \frac{1}{5}$, the following inequality holds, 
\begin{equation}
    |\frac{\sin\theta}{c_k{\theta}}| \leq \frac{1}{1-\frac{8}{3}(d\theta)^2} < \sqrt{\frac{4}{3}}.
\end{equation}
So for any $k \neq k'$, 
\begin{equation}
     | \mathbb{E}(v_k^{(pha)}v_{k'}^{(pha)})| = \frac{1}{|c_k(\theta)c_{k'}(\theta)|} |  \mathbb{E}(\text{Im}(v_k)\text{Im}(v_{k'}))|\leq \frac{4[2(d\theta)^2+2(n-2)(d\theta)]}{3N(2d-1)(\sin\theta)^2}.
\end{equation}
Also,
\begin{align}
       | \mathbb{E}((v_k^{(pha)})^2)- \frac{(2n-3)-4(d\theta)^2+4(n-2)(d\theta)}{4N(2d-1)\sin^2\theta}|  & \leq \frac{1}{c_k^2(\theta)} |\mathbb{E}(\text{Im}^2(v_k))- \frac{(2n-3)-4(d\theta)^2+4(n-2)(d\theta)}{4N(2d-1)}|\\
       &+  \frac{(2n-3)-4(d\theta)^2+4(n-2)(d\theta)}{4N(2d-1)\sin^2\theta}\frac{|c_k(\theta)^2-\sin^2\theta|}{c_k^2(\theta)}\\
       & \leq \frac{(2n-3)}{12N(2d-1)\theta^2}+\frac{7d^2}{3(2d-1)N}+\frac{3(n-2)d}{N(2d-1)\theta}.
\end{align}

Finally, we conclude the variance for the estimator \begin{equation}
    \hat{\theta} = \frac{1}{d}\sum_{k = 0}^{d-1}|c_k| \quad, \hat{\zeta} = \frac{1}{2}\frac{\vec{\mathbb{I}}^T\mathcal{D}^{-1}\vec{\Delta}}{\vec{\mathbb{I}}^T\mathcal{D}^{-1}\vec{\mathbb{I}}};
\end{equation}
where $\mathcal{D}$ is given in the Eq.~\eqref{eqn:def D}.

The variances are
\begin{equation}
    \text{Var}(\hat{\theta}) \approx \frac{2n-3}{4Nd(2d-1)} \approx \frac{n}{4Nd^2}~~~~\text{and} ~~~~\text{Var}(\hat{\zeta}) \approx \frac{3(2n-3)}{4Nd(2d-1)(d^2-1)\theta^2} \approx \frac{3n}{4Nd^4\theta^2}.
\end{equation}

\section{Cramer-Rao bound for Algorithm \ref{alg:qspe-parallel}}\label{apx:CR bound}
In this setting, the estimation is performed using $d$ repetitions with $N$ samples per circuit. Altogether, the learning procedure involves $2(2d-1)$ circuits, each with depth $\Theta(d)$. We compute the optimal variance bounded by Cramer-Rao lower bound for pre-asymptotic regime, when $d\theta <<1$, following the same analysis from \cite{dong2025qspe}.  

For the parallel learning algorithm, we simultaneously work on $(n-1)$ invariant subspaces and infer the $(\theta,\zeta)$ pair for each subspace independently. In the following discussion, we focus on the $k$th invariant subspace which can be parametrized with the same parameter set introduced in \cite{dong2025qspe}, i.e. $\Xi = (\varphi_k)=(\theta,\zeta,\chi)$ with $\chi = 0$ in our scenario. In the large sample limit, where the experimental estimated probability can be approximated by a normal distribution with variance bounded as in Eq.~\eqref{eqn:variance bound}, then the Fisher information matrix for this specific subspace is
\begin{align}
    I_{kk'}(\Xi) & = \sum_{j = 0}^{2d-2} \text{Var}^{-1}(p_{k,res,X}^M) \frac{\partial p_X^B(\omega_j,\Xi)}{\partial \varphi_k}\frac{\partial p_X^B(\omega_j,\Xi)}{\partial \varphi_{k'}}+ \sum_{j = 0}^{2d-2} \text{Var}^{-1}(p_{k,res,Y}^M) \frac{\partial p_Y^B(\omega_j,\Xi)}{\partial \varphi_k}\frac{\partial p_X^B(\omega_j,\Xi)}{\partial \varphi_{k'}}\\
     &  \approx  (\frac{4N}{(2n-3)+4(d\theta)^2+4(n-2)(d\theta))})\sum_{j = 0}^{2d-2} \frac{\partial p_X^B(\omega_j,\Xi)}{\partial \varphi_k}\frac{\partial p_X^B(\omega_j,\Xi)}{\partial \varphi_{k'}}+ \sum_{j = 0}^{2d-2}  \frac{\partial p_Y^B(\omega_j,\Xi)}{\partial \varphi_k}\frac{\partial p_X^B(\omega_j,\Xi)}{\partial \varphi_{k'}}\\
     & \approx (\frac{4N(2d-1)}{(2n-3)+4(d\theta)^2+4(n-2)(d\theta))})\text{Re}(\sum_{-d+1}^{d-1}\frac{\partial c_j(\Xi)}{\partial \varphi_k}\frac{\partial c_j(\Xi)}{\partial \varphi_{k'}}).\\
\end{align}
Following the relation $c_j(\Xi) \approx i e^{-i\chi}e^{-i(2j+1)\zeta}\theta \mathbb{I}_{j >0}$ when $d\theta <<1$, we can compute the fisher information matrix as
\begin{align}
     I(\Xi) \approx (\frac{4N(2d-1)}{(2n-3)+4(d\theta)^2+4(n-2)(d\theta))}) \begin{bmatrix}
         d & 0 & 0 \\
         0 & \frac{d(4d^2-1)\theta^2}{3}& d^2\theta^2\\
         0 &  d^2\theta^2 & d\theta^2
     \end{bmatrix}.
\end{align}

The covariance matrix of the estimator set $(\Xi)$ is bounded as follows according to Cramer-Rao lower bound,
\begin{align}
    \text{Cov}(\hat{\theta},\hat{\zeta},\hat{\chi}) \geq  I^{-1}(\Xi) \approx \frac{(2n-3)}{4Nd(2d-1)} \begin{bmatrix}
        1 & 0 & 0 \\
        0 & \frac{3}{(d^2-1)\theta^2} & \frac{-3d}{(d^2-1)\theta^2}\\
        0 & \frac{-3d}{(d^2-1)\theta^2} & \frac{4d^2-1}{(d^2-1)\theta^2}
    \end{bmatrix}.
\end{align}
As a result, the optimal variances for $(\hat{\theta},\hat{\zeta})$ are
\begin{align}
    & \text{Var}_{opt}(\hat{\theta}) \approx \frac{n}{4Nd^2}; ~~~~~~~~\text{Var}_{opt}(\hat{\zeta}) \approx \frac{3n}
     {4Nd^4\theta^2}. 
\end{align}

\section{Total evolution time scaling comparison}\label{apx:total time}
The total evolution time of the whole learning algorithm is defined as 
\begin{equation}
    T =\sum_{i = 1}^{r} (\tau_i \times N),
\end{equation}
where $r$ is the total number of experiment rounds, $\tau_i$ is the running time for each experiment, and $N$ is the number of samples taken for each experiment. For analog-digital-hybrid learning algorithm \ref{alg:qspe-parallel} in which all the invariant subspaces are learned in parallel, we need to run $r = (n-1)$ experiments with each one taking $\tau_i = 2(2d-1)\times d $ running time and $N$ samples for each circuit. Thus, the total evolution time can be computed as follows.
\begin{theorem*}[Upper bound for the total evolution time for algorithm \ref{alg:qspe-parallel}]\label{thm:upper bound-parallel}
  Assuming the quantum system is evolving under a $2$-body Hamiltonian specifying as Hamiltonian Eq.~\eqref{eqn: all2all H}, our algorithm provides estimators $\mathbf{\hc}$ with probability at least $1-\delta$ satisfying $|{\hc}_{ij}-{c}_{ij}| \leq \epsilon$ for all $i,j \in \{1,\cdots,n\}$ and $i \neq j$ with total evolution time $T = O(\frac{n^{3/2}}{\epsilon}).$
\end{theorem*}
\begin{proof}
According to \text{Chebyshev equality}, for any individual estimator $\hc_{ij}$, 
\begin{align}
    \Pr(|\hc_{ij}-c_{ij}|\leq \epsilon) \geq 1-\frac{\text{Var}(\hc_{ij})}{\epsilon^2}.
\end{align}
We plug in the variance given in Eq.~\eqref{eqn:var_c_parallel} and set $\frac{\text{Var}(\hc_{ij})}{\epsilon^2} = \delta$, we get $d^4 = \frac{3n}{4Na_i^2\epsilon^2\delta}$, then the total evolution time $T = \sum_{i = 1}^r (\tau_i \times N) = (n-1)\times (2(2d-1)\times d)\times N  \approx nd^2N$ scales as 
\begin{equation}
    T = O(\frac{n^{3/2}}{\epsilon}).
\end{equation}
\end{proof}
We have the quadratic speedup result which indicates $O(n)$ experimental rounds are required for learning $O(n^2)$ terms as indicated in Theorem \ref{thm:quadratic speedup}, but the total evolution time scales with $O(n^{3/2})$ which is due to the sampling overhead for parallelizing $O(n)$ invariant subspaces. 

We also consider here applying the approach proposed in \cite{huang2023learning} for learning all-to-all connected Hamiltonians. In this method, the many-body Hamiltonian is first decoupled into non-interacting patches by applying random Pauli operations on the qubits that connect different patches. Once decoupled, the remaining non-interacting pairs can be learned in parallel using robust phase estimation. However, in the all-to-all connected case, the Hamiltonian is no longer geometrically local. To fully disconnect the interaction graph, one must eliminate all $(n-2)$ residual couplings, leaving only a single interaction active at a time. This requires $O(n^2)$ experimental rounds, i.e., $r = O(n^2)$. Each experiment achieves Heisenberg-limited precision, with $\tau \times N \sim O(\epsilon^{-1})$, so the total evolution time scales as $O(n^2 / \epsilon)$. Thus, the total evolution time of our parallel learning algorithm still scales better than the $O(n^2)$ suggests in \cite{huang2023learning} where $O(n^2)$ experimental rounds are required for learning all-to-all connected Hamiltonian for $n$ qubits.

Besides, for the fully-analog learning algorithm \ref{alg:qspe-analog}, we increase the experimental rounds to focus on each invariant subspace one at a time for ensuring the whole algorithm is implementable on a fully analog quantum simulator. So here the number of rounds is $r = \frac{n(n-1)}{2}$ with each one taking $\tau_i = 2(2d-1)\times d$ running time and $N$ samples for each circuit. 

\begin{theorem}[Upper bound for the total evolution time for algorithm \ref{alg:qspe-analog}]\label{thm:upper bound_analog}
    Assuming the quantum system is evolving under a $2$-body Hamiltonian specifying as Hamiltonian Eq.~\eqref{eqn: all2all H}, our algorithm provides estimators $\mathbf{\hc}$ with probability at least $1-\delta$ satisfying $|{\hc}_{ij}-{c}_{ij}| \leq \epsilon$ for each $i,j \in \{1,\cdots,n\}$ and $i \neq j$ with total evolution time $T = O(\frac{n^{2}}{\epsilon}).$
\end{theorem}
\begin{proof}
According to \text{Chebyshev equality}, for any individual estimator $\hc_{ij}$, 
\begin{align}
    \Pr(|\hc_{ij}-c_{ij}|\leq \epsilon) \geq 1-\frac{\text{Var}(\hc_{ij})}{\epsilon^2}.
\end{align}
We plug in the variance given in Eq.~\eqref{eqn:var_c_parallel} and set $\frac{\text{Var}(\hc_{ij})}{\epsilon^2} = \delta$, we get $d^4 = \frac{3}{8Na_i^2\epsilon^2\delta}$, then the total evolution time $T = \sum_{i = 1}^r (\tau_i \times N) = (n-1)n/2\times (2(2d-1)\times d)\times N  \approx n^2d^2N$ scales as 
\begin{equation}
    T = O(\frac{n^{2}}{\epsilon}).
\end{equation}
\end{proof}
This scaling is consistent with the total evolution time reported in~\cite{huang2023learning}, as both approaches adopt a sequential strategy. The key distinction, however, is that our learning algorithm supports in-situ learning and operates in a fully analog manner, whereas the method in~\cite{huang2023learning} relies on digital single-qubit gates and reshapes the Hamiltonian, thereby precluding in-situ learning.

\end{document}